\documentclass[onecolumn,a4paper]{article}%
\usepackage{amsmath}
\usepackage{amsfonts}
\usepackage{amssymb}
\usepackage{graphicx}%
\providecommand{\U}[1]{\protect\rule{.1in}{.1in}}
\newtheorem{theorem}{Theorem}

\newtheorem{lemma}[theorem]{Lemma}

\newtheorem{remark}[theorem]{Remark}

\newenvironment{proof}[1][Proof]{\noindent\textbf{#1.} }{\ \rule{0.5em}{0.5em}}
\begin{document}

\title{More on Intractability of Thermalization: (almost) i.i.d. inputs and finite
lattices }
\author{Keiji Matsumoto keiji@nii.ac.jp\\NII}
\maketitle

\section*{Abstract}

This work is an extention of Shiraishi and Matsumoto\cite{ShiraishiMatsumoto},
and discueese the computational complexity of the long-term average of local
observables in one-dimensional lattices with shift-invariant nearest-neighbor
interactions for simple initial states. 

As shown in the previous paper, the problem is generally intractable. In this
paper we refine the statement further. First, we consider restriction of the
initial state, where the state of all the sites are the same except for a
single site. We show this version of the problem is also undecidable
(RE-complete). Then we turn to the case where the lattice size is finite:
depening on the defitiniton of the input size, this version of problem is
either EXPSPACE-complete or PSPACE-complete.

\section{What is this paper about}

\subsection{The sketch of the problem and the results}

This work is an extension of Shiraishi and Matsumoto\cite{ShiraishiMatsumoto}.

In this paper, we discuss the computational complexity of the long-term
average of local observables in one-dimensional lattices such that: The
Hamiltonian is shift-invariant and consists of nearest neighbor interactions,
and the initial states are either
\begin{equation}
\cdots\otimes\left\vert \psi\right\rangle \otimes\left\vert \psi\right\rangle
\otimes\left\vert e_{0}\right\rangle \otimes\left\vert \psi\right\rangle
\otimes\left\vert \psi\right\rangle \otimes\cdots\label{rho-L-0}%
\end{equation}
or
\begin{equation}
\cdots\otimes\left\vert \psi\right\rangle \otimes\left\vert \psi\right\rangle
\otimes\left\vert \psi\right\rangle \otimes\left\vert \psi\right\rangle
\otimes\left\vert \psi\right\rangle \otimes\cdots. \label{rho-L-iid-0}%
\end{equation}
Some more additional technical conditions are introduced so that the
single-site version of the problem is easily solved by a Turing machine. We
show that both of them are intractable, and it seems that the former is much
more intractable (undecidable, indeed) than the latter.

For the statement of our result, \ the detail of the setting is explained.
Throughout the paper, we fix a CONS of the Hilbert space of each lattice, and
denote it $\{\left\vert e_{\kappa}\right\rangle \}_{\kappa=0}^{d-1}$. First,
we suppose the components of $\left\vert \psi\right\rangle $, $\left\vert
e_{0}\right\rangle $, or terms of Hamiltonian are "easily computable" by a
Turing machine from bit strings specifying them. For example, in the models
that is used to show the intractability of these problems, they are computed
by applying finitely many $+$,$-$, $\times$, $/$ and $\sqrt{}$ to the natural
numbers that is encoded by the bit string. The word "easily computable" means
that the time for computing an approximation with the error not more than
$\varepsilon$ is bounded by a polynomial function of $\log\varepsilon$ and the
number of digits of the input bit string.

Second assumption is that gaps between two distinct energy levels (each level
may be degenerate) are bounded from below by a rational number that can be
easily computed provided the lattice size is finite.

If the components of

For the statement of our result, \ the detail of the setting is explained.
Throughout the paper, we fix a CONS of the Hilbert space of each lattice , and
denote it by $\{\left\vert e_{\kappa}\right\rangle \}_{\kappa=0}^{d-1}$.
\bigskip Each term of the Hamiltonian and the state vector $\left\vert
\psi\right\rangle ,\,\left\vert e_{0}\right\rangle $ are specified by their
components with respect to a fixed basis .

Discussion of computation of continuous quantities by Turing machine, which
can manage only finitely many bits, is rather subtle.

\ \ 

Moreover, gaps between two distinct energy levels (each energy level may
degenerate) are bounded by The latter assumption is equivalent to that an
upper bound to the relaxation time $T$ is easily computed for each finite
lattice size.

Therefore, our problem is tractable if the lattice size is small, but

Without these condition, clearly the problems are intractable even for a
single system: For example, consider a two level system where the energy gap
is arbitrary real number and identity between the two energy levels are
decidable by no TM (This is the case even if they are arbitrary computable
real numbers.). Then clearly the long-term average is also impossible to compute.

\begin{remark}
The celebrated theorem\cite{CGW}\cite{BCLP} stating undecidability of a
spectral gap in a infinite-size lattice is not relevant here. In our setting,
the time duration goes to infinite before the lattice size is taken to
infinity, so the spectral gap of a finite-size lattice is relevant.
\end{remark}

In most part of the paper, we consider the Hamiltonian as a parameter of the
problem, and the state $\left\vert \psi\right\rangle $ is the input of the
problem ($\left\vert e_{0}\right\rangle $ can be fixed without loss of
generality). This means that the Turing machine (program) solving the problem
may depend on the Hamiltonian but must be independent of $\left\vert
\psi\right\rangle $. (We also discuss the version in which the state is a
parameter and the Hamiltonian is the input.) The state is represented by a bit
string $v$, and the components of $\left\vert \psi\right\rangle $ is computed
from $v$ by a TM with the error at most $\varepsilon$ using time that is a
polynomially bounded function of $n:=|v|$ and $\log(1/\varepsilon)$. Moreover,
they are algebraic numbers of degree $2^{p(n)}$ (Here $p(n)$ is a polynomially
bounded function of $n$) in the examples we use for the proof of the hardness.

The output of the problem is also discretized: we question whether the
long-term average is above a certain threshold or not. Then the problem is
undecidable (cannot be solved by any TM) if the initial state is in the form
of (\ref{rho-L-0}). If the initial state is (\ref{rho-L-iid-0}), it seems that
the problem becomes easier. First, it is \textsf{EXPSPACE}-hard (at least as
hard as any problems with space $2^{p(n)}$ (Here $p(n)$ is a polynomially
bounded function of $n$).Moreover, if the lattice size is finite and
$O(2^{p(n)})$, then the problem can be solved using space $2^{q(n)}$ ( $p(n)$
and $q(n)$ are polynomially bounded functions of $n$), it is contained in
\textsf{EXPSPACE}.

These assertions are proved for a space average of a single site observable
$A$, where $A$ is almost arbitrary. So in fact these statements can be
recasted in terms of the space average of the single-site density operator.
Indeed, we prove the intractability of the long-term average by reducing to it
to the following decision problem: we question weather the space average of
the single-site density operator stays in the neighborhood of $\left\vert
e_{1}\right\rangle $ forever or eventually comes close to the mixed state
$\frac{1}{2}(\left\vert e_{1}\right\rangle \left\langle e_{1}\right\vert
+\left\vert e_{2}\right\rangle \left\langle e_{2}\right\vert )$. This means
that the effect of small perturbation in $\left\vert \psi\right\rangle
\approx\left\vert e_{1}\right\rangle $ is quite unpredictable. (Here, the
small perturbation is added to all the sites, so it is in fact quite large as
a whole: The problem is whether its effect can be observed locally or not.)

To show the intractability of this problem on the single-site density
operator, we reduce to it the halting problem of Turing machines
(\textsf{Halt}, hereafter), which is a textbook example of an undecidable
problem: Given a description of a TM and an input to it, the question is
whether the computation starting from the input ever terminates or not. This
question is equivalent to whether a universal Turing machine (UTM), a Turing
machine that simulates any TM, halts on a given input. The proof exploits the
correspondence between a Hamiltonian dynamics (Hamiltonian cell automata) and
a reversible discrete time dynamics by \cite{NagajWojcan}. The former is not
a(n approximation to) continuous incorporation of the time development of the
latter, but eventually the `probability' of staying a state of a time step of
the latter becomes almost uniform. By this powerful theoretical tool, we can
mimic a UTM if it is reversible (URTM, in short \cite{Morita}). Therefore,
tracking the dynamics of URTM cannot be easier than tracking the corresponding
Hamiltonian dynamics. Since the former is intractable, so is latter.

However, there are several subtle points. First, we have to detect whether the
URTM $M$ has halted or not by a space average of a single site observable $A$,
and $A$ is almost arbitrary. Second, initially all the sites are in the same
state $\left\vert \psi\right\rangle $ except perhaps the $0$-th site, so the
informations that regulates the dynamics including the input to the URTM $M$
are encoded into $\left\vert \psi\right\rangle $. \ Since the dimension of the
single site is fixed and the size of the input to $M$ is arbitrary, the
informations are encoded to the components of $\left\vert \psi\right\rangle $.
So the decode of the information is not very trivial. Here we cannot rely on
the phase estimation as in \cite{NagajWojcan}, since the information is
encoded to a state and not to a Hamiltonian. In addition, the initial state is
necessarily in superposition of various classical configurations. However, in
\cite{NagajWojcan} they deals with the dynamics starting from a state
corresponding to a single classical configuration, so we have to evaluate the
effect of the interference between them.

Major part of the paper is devoted to circumvent these difficulties.

In addition, there are some more minor results. First, if the initial state is
(\ref{rho-L-0}), the problem is as difficult as the halting problem
\textsf{Halt}, in the sense that the problem can be reduced to \textsf{Halt}.
(The reduction is either Turing reduction or many-one reduction, depending on
the setting.) Second, we present the version of the problem where the input is
Hamiltonian and the initial state is a parameter, and showed these two
versions are equivalent. Third, when the lattice size is finite and $O(p(n))$
($p$ is a polynomial) and the input is the Hamiltonian, the problem is
\textsf{PSPACE}-complete. In showing the last result, we use a version of
phase estimation that can be emulated by the Hamiltonian cell automata without
a clock counter.

In \cite{BCLP} and \cite{CGW}, the undecidability of the spectral gap is
proved, and in the course of the proof, they proved the undecidability of
whether the ground energy is larger than a threshold or not. This statement
can be read as undecidability of the long-term average of the energy when the
initial state is the ground state. Though they don't state this explicitly, we
call this statement as "theirs", since it is trivial given their results.

There are some strong points in our result compared with "theirs". First,
while the initial state in the "their" version is quite complicated, our
initial state is in a much simpler form.

Second, in our case, the Hamiltonian of the system is fixed, so it is
incomputably is true even if \ we use the algorithm depending on the
Hamiltonian. Meantime, in "their" setting, both the initial state and the
Hamiltonian are the input of the problem, so the claim is weaker than ours.

Third, in our case, the single site observable can be almost arbitrary, which
is not the case for "their" version, since the ground state strongly depends
on the input string.

Fourth, the dimension of each site is probably smaller than theirs: First, in
their construction each site is tensor product of the Hilbert space
corresponding to a tape cell and a finite control, respectively. Meantime, in
our case, each site is the direct sum of these two spaces. Second, we do not
use a clock counter to regulate the move of the machine, which is used in
\cite{BCLP} and \cite{CGW}.

\subsection{Statement of the problem and the main theorems}

Consider 1-dim chain of $d$-level quantum system, $\otimes_{i=-\infty}%
^{\infty}\mathcal{H}_{i}$, $\dim\mathcal{H}_{i}=d$, with the shift-invariant
Hamiltonian $H$, only with nearest neighbor interaction. Our interest is the
space average of the single site observable $A$, and its long-term average.
The purpose of the paper is to show the computation of the long-term average
is impossible by any Turing machine.

We show the assertion is true even if the initial state is as simple as
(\ref{rho-L-0}), where $\,\left\vert \psi\right\rangle $ and $\left\vert
e_{0}\right\rangle $ are mutually orthogonal pure states, $\,\left\langle
e_{0}\right\vert \left.  \psi\right\rangle =0$. Moreover, the single site
observable $A$ can be arbitrary operator which is not trivial, i.e., a
constant multiple of $I$.

We define and compute every quantity on the finite size cluster with $L+1$
sites, and then take $L$ to $\infty$ in the end. Though we state our arguments
in the periodic boundary condition, they are generalized to the open boundary
condition without much difficulty.

In the case of periodic boundary condition, the $(L+1)$-th site is identical
with the $0$-th site. So, consider $\otimes_{i=0}^{L}\mathcal{H}_{i}$ , and%
\begin{equation}
\rho^{L}=\left\vert e_{0}\right\rangle \left\langle e_{0}\right\vert
\otimes(\left\vert \psi\right\rangle \left\langle \psi\right\vert )^{\otimes
L}, \label{rho-L}%
\end{equation}
corresponding to (\ref{rho-L-0}) and
\begin{equation}
\rho^{L}=(\left\vert \psi\right\rangle \left\langle \psi\right\vert )^{\otimes
L}, \label{rho-L-iid}%
\end{equation}
corresponding to (\ref{rho-L-iid-0}).

Let us denote the `restriction' of $H$ to $\otimes_{i=0}^{L}\mathcal{H}_{i}$
by $H^{L}$, by identifying the $L+2\,$-th site with the $0$-th.

Let $A$ be an observable in $\mathcal{H}$, and our interest is the space
average of $A$'s:
\[
A^{(L)}:=\frac{1}{L+1}\sum_{i=0}^{L}A_{i},
\]
where $A_{0}:=A\otimes I\otimes I\cdots\otimes I$, $A_{1}:=I\otimes A\otimes
I\otimes\cdots\otimes I$, etc.. Equivalently, we are interested in the space
average of \ the state,
\[
\overline{\rho}(t,L):=\frac{1}{L+1}\sum_{i=0}^{L}\rho_{i}(t),
\]
where%

\begin{align*}
\rho^{L}(t)  &  :=e^{iH^{L}}\rho^{L}e^{-iH^{L}},\\
\rho_{i}(t)  &  :=\,\mathrm{tr}\,_{\otimes_{i^{\prime}:i^{\prime}\neq
i}\mathcal{H}_{i^{\prime}}}\rho(t)=\left.  \rho(t)\right\vert _{\mathcal{H}%
_{i}}.
\end{align*}

We argue, roughly, that
\begin{align*}
&  \lim_{L\rightarrow\infty}\lim_{T\rightarrow\infty}\frac{1}{T}\int_{0}%
^{T}\mathrm{tr}\,\rho^{L}(t)\,A^{(L)}dt\\
&  =\,\lim_{L\rightarrow\infty}\lim_{T\rightarrow\infty}\frac{1}{T}\int
_{0}^{T}\mathrm{tr}\,\overline{\rho}(t,L)\,Adt
\end{align*}
cannot be computed.

For simplicity, we consider the decision theory version of the problem.\ A
decision problem is a problem that can be answered either Yes or No. It is
said to be \textit{decidable} if there is a TM that halts on any input and can
solve it, and \textit{undecidable} if not\thinspace\cite{BCLP}. Whenever we
say something is an input of the problem, the Turing machine solving the
problem should be independent of it: We should find a program that solves the
problem for all the presupposed inputs. On the other hand, if something is
constant or fixed, then the Turing machine solving the problem may vary with it.

In our case, we consider the situation either
\begin{equation}
\,\,\varliminf_{L\rightarrow\infty}\,\left\vert \lim_{T\rightarrow\infty}%
\frac{1}{T}\int_{0}^{T}\,\mathrm{tr}\,\overline{\rho}(t,L)\,A\,dt\,-c_{1}%
\right\vert \,\,\geq2\varepsilon_{0} \label{A>c}%
\end{equation}
or
\begin{equation}
\varlimsup_{L\rightarrow\infty}\left\vert \lim_{T\rightarrow\infty}\frac{1}%
{T}\int_{0}^{T}\,\mathrm{tr}\,\overline{\rho}(t,L)\,A\,dt\,-c_{1}\right\vert
\,\,\leq\varepsilon_{0}. \label{A<c}%
\end{equation}
is correct, and question whether the former is true or not. Clearly, this
decision problem version is not more difficult than the computation of the
limit. Also one can obtain the approximation of the quantity by asking this
question for various values of $c_{1}$, provided the limit exists. Here,
$\lim_{T\rightarrow\infty}$ exists since the system is finite dimensional for
each $L$, while $\lim_{L\rightarrow\infty}$ may not exist. \ 

Roughly, there are two ways of defining the problem: either fix the
Hamiltonian and takes the initial state as an input, or the other way around.
The former is called \textsf{OAS}$(d,A,H,f_{\psi},c_{1},\varepsilon_{0})$:

\begin{description}
\item[{$\mathbf{[[}\mathsf{OAS}(d,A,H,f_{\psi},c_{1},\varepsilon
_{0})\mathbf{]]}$}] 

\item[Fixed:] $d:=\dim\mathcal{H}$, $A$, $c_{1}$, $\varepsilon_{0}$, 1- and 2-
body terms of $H$, $\left\vert e_{0}\right\rangle $, the function $f_{\psi}$
of a bit string to a state vector $f_{\psi}(v)=\left\vert \psi\right\rangle
\in\mathcal{H}$: They are chosen so that either (\ref{A>c}) or (\ref{A<c}) is
true for the initial state (\ref{rho-L}) for any input $v$.

\item[Input:] A bit string $v$.

\item[Question:] Whether (\ref{A>c}) is the case or not when the initial state
is (\ref{rho-L}).
\end{description}

If the Hamiltonian is the input, the problem is called \textsf{OAH}%
$(d,A,f_{H},\left\vert \psi\right\rangle ,c_{1},\varepsilon_{0})$:

\begin{description}
\item[{$\mathbf{[[}\mathsf{OAH}(d,A,\left\vert \psi\right\rangle ,f_{H}%
,c_{1},\varepsilon_{0})\mathbf{]]}$}] 

\item[Fixed] $d:=\dim\mathcal{H}$, $A$, $c_{1}$, $\varepsilon_{0}$,
$\left\vert \psi\right\rangle $, $\left\vert e_{0}\right\rangle $, the
function $f_{H}$ of a bit string to 1- and 2- body terms of $H$. : They are
chosen so that either (\ref{A>c}) or (\ref{A<c}) is true for the initial state
(\ref{rho-L}) for any input $v$.
\end{description}

When the initial state is (\ref{rho-L-iid}), we replace the eq. (\ref{rho-L})
by the eq. (\ref{rho-L-iid}), and remove $\left\vert e_{0}\right\rangle $ from
the list of the fixed objects: the modification of \textsf{OAS} and
\textsf{OAH} in this way is denoted by \textsf{OAS-iid }and \textsf{OAH-iid}, respectively.

In the description of the problems, linear operators and vectors are
represented by components in terms of a standard basis, and we suppose
$\left\vert e_{0}\right\rangle $ is one of basis vectors without loss of generality.

Since our interest is in the difficulty of the problem that arises from
composition of subsystems, we put the assumptions so that the single-site
version of the problem is tractable. First, the states and the Hamiltonians
are easily computable function of natural numbers: In fact, we can even
restrict the computation to composition of finite numbers of arithmetic
operations ($+,-,/,\ast$)\ and $\sqrt{\cdot}$. Second, we suppose the gaps
between two distinct energy levels are easily computable.

\begin{theorem}
\label{thm:OA}Suppose $d=\dim$ $\mathcal{H}\geq d_{0}$. Suppose $f_{\psi}$ and
$f_{H}$ can be computed with the error at most $\varepsilon$ using time which
is bounded from above by a polynomial in $n:=\left\vert v\right\vert $ and
$\log(1/\varepsilon)$. Moreover, $f_{\mathrm{gap}}(L,n)$ is polynomial-time
computable. Moreover, if $\lambda$ and $\lambda^{\prime}$ are two distinct
eigenvalues of $H$,
\begin{equation}
\left\vert \lambda-\lambda^{\prime}\right\vert \geq1/f_{\mathrm{gap}}(L,n).
\label{gap}%
\end{equation}
Here $f_{\mathrm{gap}}(L,n)$ is a polynomial\thinspace-\thinspace time
computable integer valued function. Then

(i) There is an $H$ satisfying the assumption such that: for most of $A$,
i.e., any $A$ with $\left\langle e_{1}\right\vert A\left\vert e_{1}%
\right\rangle \neq\left\langle e_{2}\right\vert A\left\vert e_{2}\right\rangle
$ where $\{\left\vert e_{\kappa}\right\rangle \}_{\kappa=0,1,2}$ is an
orthonormal set of vectors, \textsf{OAS(-iid)}$(d,A,H,f_{\psi},c_{1}%
,\varepsilon_{0})$ for some $c_{1}$, $\varepsilon_{0}$ is undecidable
(\textsf{EXPSPACE}-hard).

(ii) For any $\left\vert e_{0}\right\rangle $ and $\left\vert \psi
\right\rangle $ with $\left\langle e_{0}\right.  \left\vert \psi\right\rangle
=0$ and for most of $A$, i.e., any $A$ which is not constant multiple of the
identity on the orthocomplement subspace of $\left\vert e_{0}\right\rangle $,
\textsf{OAH(-iid)}$(d,A,f_{H},\left\vert \psi\right\rangle ,c_{1}%
,\varepsilon_{0})$ for some $c_{1}$, $\varepsilon_{0}$ is undecidable
(\textsf{EXPSPACE}-hard).
\end{theorem}

Here, \textsf{EXPSPACE} is a set of decision problems that can be solved by a
TM using $S(n)$ tape cells, where $S(n)=\exp(p(n))$ for some polynomial
$p(n)$. A decision problem is \textsf{EXPSPACE}-hard if any member of
\textsf{EXPSPACE} can be converted to that problem efficiently, \textit{i.e.},
the time needed for the conversion is in polynomial of $n$ (many-one
reduction). \ So the problem seems to be easier for initial states in the form
of (\ref{rho-L-iid}), although it is still awfully difficult.

Observe the observable $A$ is almost arbitrary in the statement of the above
theorems. This indicates that our argument is concerned essentially with a
state rather than an observable To make this point explicit, let us consider
an decision problem about the space average of the single site state. We
suppose the long-term behavior of $\overline{\rho}(t,L)$ is either one of
them: Its long-term average becomes a specific mixed state
\begin{equation}
\varlimsup_{L\rightarrow\infty}\left\Vert \lim_{T\rightarrow\infty}\frac{1}%
{T}\int_{0}^{T}\overline{\rho}(t,L)\mathrm{d}t-((1-\eta)\left\vert
e_{1}\right\rangle \left\langle e_{1}\right\vert +\eta\left\vert
e_{2}\right\rangle \left\langle e_{2}\right\vert )\right\Vert _{1}%
\leq\varepsilon_{1},\label{sa-1}%
\end{equation}
or stays in the neighbor of $\left\vert e_{1}\right\rangle $ for all the time
\begin{equation}
\forall L\geq L_{0},\,\forall T\in\lbrack0,\infty]\left\Vert \frac{1}{T}%
\int_{0}^{T}\overline{\rho}(t,L)\mathrm{d}t-\left\vert e_{1}\right\rangle
\left\langle e_{1}\right\vert \right\Vert _{1}\leq\varepsilon_{1},\label{sa-2}%
\end{equation}
Here $L_{0}$ is function of $\left\vert \psi\right\rangle $ and $\varepsilon
_{1}$ that can be computed by a Turing machine that always halts,
\textit{i.e}., a total recursive function.

Clearly,
\begin{equation}
\forall L\geq L_{0},\,\forall t\in\lbrack0,\infty],\,\,\left\Vert
\overline{\rho}(t,L)-\left\vert e_{1}\right\rangle \left\langle e_{1}%
\right\vert \right\Vert _{1}\leq\varepsilon_{1}. \label{sa-2-2}%
\end{equation}
is sufficient for the latter to hold.

When the input is the initial state, we call the problem $\mathsf{SAS}%
(d,H,f_{\psi},\eta,\varepsilon_{1})$, which is defined by:

\begin{description}
\item[{$\mathbf{[[}\mathsf{SAS}(d,H,f_{\psi},\eta,\varepsilon_{1})\mathbf{]]}%
$}] 

\item[Fixed:] $d=\dim\mathcal{H}$, $\{\left\vert e_{\kappa}\right\rangle
;\left\langle e_{\kappa^{\prime}}\right.  \left\vert e_{\kappa}\right\rangle
=\delta_{\kappa,\kappa^{\prime}}\}_{\kappa=0,1,2}$, 1- and 2- body terms of
$H$,$\ \eta$ and $\varepsilon_{1}$ with $0<2\varepsilon_{1}<\eta<1$, the
function $f_{\psi}$ of a bit string to a state vector $f_{\psi}(v)=\left\vert
\psi\right\rangle $ living in $\mathcal{H}$: They are controlled so that:
$\overline{\rho}(0,\infty)=\left\vert \psi\right\rangle \left\langle
\psi\right\vert $ is close to $\left\vert e_{1}\right\rangle \left\langle
e_{1}\right\vert $,
\begin{equation}
\,\left\Vert \left\vert \psi\right\rangle \left\langle \psi\right\vert
-\left\vert e_{1}\right\rangle \left\langle e_{1}\right\vert \right\Vert
_{1}\leq\varepsilon_{1}, \label{sa-0}%
\end{equation}
and either (\ref{sa-1}) or (\ref{sa-2}) is true for the initial state
(\ref{rho-L}).

\item[Input:] A bit string $v$.

\item[Question:] (\ref{sa-1}) is true or not when the input is (\ref{rho-L}).
\end{description}

Meantime, if the input is the Hamiltonian, the problem is called
$\mathsf{SAH}(d,f_{H},\eta,\varepsilon_{1})$, and the list of the constants is changed:

\begin{description}
\item[{$\mathbf{[[}\mathsf{SAS}(d,H,f_{\psi},\eta,\varepsilon_{1})\mathbf{]]}%
$}] 

\item[Fixed:] $d=\dim\mathcal{H}$, $\{\left\vert e_{\kappa}\right\rangle
;\left\langle e_{\kappa^{\prime}}\right.  \left\vert e_{\kappa}\right\rangle
=\delta_{\kappa,\kappa^{\prime}}\}_{\kappa=0,1,2}$, $\left\vert \psi
\right\rangle :=\left\vert e_{1}\right\rangle $,$\ \eta$ and $\varepsilon_{1}$
with $0<2\varepsilon_{1}<\eta<1$, the function $f_{H}$ of $v$ to 1- and 2-body
terms of the Hamiltonian $H$. They are controlled so that either (\ref{sa-1})
or (\ref{sa-2}) is true for the initial state (\ref{rho-L}).
\end{description}

When the initial state is (\ref{rho-L-iid}), we replace the eq. (\ref{rho-L})
by (\ref{rho-L-iid}), and remove $\left\vert e_{0}\right\rangle $ from the
list of the fixed objects: the modification of \textsf{SAS} and \textsf{SAH}
in this way is denoted by \textsf{SAS-iid }and \textsf{SAH-iid}, respectively.

Here note that the whole state $\rho^{L}$ can significantly varies with the
input, though the average single-site state $\overline{\rho}(0,\infty)$ is
almost constant of it.

\begin{theorem}
\label{thm:SA}Suppose $d:=\dim$ $\mathcal{H}\geq d_{0}$ . Suppose also
$f_{\psi}$ and $f_{H}$ can be computed with the error at most $\varepsilon$
using time in polynomial of $n:=\left\vert v\right\vert $ and $\log
(1/\varepsilon)$. Moreover, the energy gaps of the Hamiltonian $H$ satisfy
(\ref{gap}). Then

(i) For any $\{\left\vert e_{\kappa}\right\rangle \}_{\kappa=0,1,2}$ there is
an $H$ such that \textsf{SAS(-iid)}$(d,H,f_{\psi},\frac{1}{2},\varepsilon
_{1})$ with $\varepsilon_{1}\in(0,1/4)$ is undecidable (\textsf{EXPSPACE}-hard).

(ii) For any $\{\left\vert e_{\kappa}\right\rangle \}_{\kappa=0,1,2}$
\textsf{SAH(-iid)}$(d,f_{H},\frac{1}{2},\varepsilon_{1})$ with $\varepsilon
_{1}\in(0,1/4)$ is undecidable (\textsf{EXPSPACE}-hard).
\end{theorem}

In fact, $\mathsf{SAS}(d,H,f_{\psi},\eta,\varepsilon_{1})$ and $\mathsf{SAH}%
(d,f_{H},\eta,\varepsilon_{1})$ turns out to be RE-complete. Also, the former
and the latter \textsf{is} reduced to, therefore not harder than,
$\mathsf{OAS}(d,A,H,f_{\psi},c_{1},\varepsilon_{0})$ and $\mathsf{OAH}%
(d,A,f_{H},\left\vert \psi\right\rangle ,c_{1},\varepsilon_{0})$,
respectively, so Theorems \ref{thm:OA} is a corollary of Theorem \ref{thm:SA}.

\subsection{The main lemmas and the proof of the theorems}

\noindent\textbf{Proof of Theorems\thinspace\ref{thm:OA} from Theorem
\ref{thm:SA}.} We only show the item (i) of them for $\mathsf{OAS}%
(d,A,H,f_{\psi},c_{1},\varepsilon_{0})$, as the proof for the other cases are
almost analogous.

Let the Hamiltonian $H$ be as of the Theorem\thinspace\ \ref{thm:SA}, (i).
Then the instance of SAS with small $\varepsilon_{1}>0$ can be solved by
observing the the long-term average of the space average of $A$ with
$\left\langle e_{1}\right\vert A\left\vert e_{1}\right\rangle $ $\neq
\left\langle e_{2}\right\vert A\left\vert e_{2}\right\rangle $, so it is
reduced to an instance of OAS with certain parameters. This instance of OAS is
not easier than the corresponding instance of SAS, so Theorem \ref{thm:OA},
(i) follows from Theorem \ref{thm:SA}, (i). The parameters of the instance of
SAS and OAS are%
\begin{align}
c_{1}  &  :=\left\langle e_{1}\right\vert A\left\vert e_{1}\right\rangle
,\nonumber\\
\varepsilon_{0}  &  :=\varepsilon_{1}\left\Vert A\right\Vert =\frac{1}{3}%
\eta|\left\langle e_{1}\right\vert A\left\vert e_{1}\right\rangle
-\left\langle e_{2}\right\vert A\left\vert e_{2}\right\rangle |\nonumber\\
\varepsilon_{1}  &  :=\frac{1}{3\left\Vert A\right\Vert }\eta|\left\langle
e_{1}\right\vert A\left\vert e_{1}\right\rangle -\left\langle e_{2}\right\vert
A\left\vert e_{2}\right\rangle |, \label{parameters}%
\end{align}
so that, for any $\sigma_{1}$ with $\left\Vert \sigma_{1}-\left\vert
e_{1}\right\rangle \left\langle e_{1}\right\vert \right\Vert _{1}%
\leq\varepsilon_{1}$ and any $\sigma_{2}$ with $\left\Vert \sigma_{2}%
-((1-\eta)\left\vert e_{1}\right\rangle \left\langle e_{1}\right\vert
+\eta\left\vert e_{2}\right\rangle \left\langle e_{2}\right\vert )\right\Vert
_{1}\leq\varepsilon_{1},$
\begin{align*}
\,\left\vert \mathrm{tr}\,\sigma_{2}A-c_{1}\right\vert  &  \leq\left\Vert
A\right\Vert \left\Vert \sigma_{2}-((1-\eta)\left\vert e_{1}\right\rangle
\left\langle e_{1}\right\vert +\eta\left\vert e_{2}\right\rangle \left\langle
e_{2}\right\vert )\right\Vert _{1}\\
&  \leq\varepsilon_{1}\left\Vert A\right\Vert \leq\varepsilon_{0},\,\\
\,\left\vert \mathrm{tr}\,A(\sigma_{1}-\sigma_{2})\right\vert  &
\geq\,|\mathrm{tr}\,A\left\vert e_{1}\right\rangle \left\langle e_{1}%
\right\vert -\,\mathrm{tr}\,A((1-\eta)\left\vert e_{1}\right\rangle
\left\langle e_{1}\right\vert +\eta\left\vert e_{2}\right\rangle \left\langle
e_{2}\right\vert )|\\
&  -\,\left\Vert A\right\Vert \left\Vert \sigma_{1}-\left\vert e_{1}%
\right\rangle \left\langle e_{1}\right\vert \right\Vert _{1}\\
&  -\left\Vert A\right\Vert \left\Vert \sigma_{2}-((1-\eta)\left\vert
e_{1}\right\rangle \left\langle e_{1}\right\vert +\eta\left\vert
e_{2}\right\rangle \left\langle e_{2}\right\vert )\right\Vert _{1}\\
&  \geq2\eta|\left\langle e_{1}\right\vert A\left\vert e_{1}\right\rangle
-\left\langle e_{2}\right\vert A\left\vert e_{2}\right\rangle |-2\left\Vert
A\right\Vert \varepsilon_{1}\\
&  \geq\varepsilon_{0}.
\end{align*}
$\blacksquare$

In showing Theorem \ref{thm:SA} for the initial state (\ref{rho-L}), we
construct a Hamiltonian and a map of the bit string $v$ to the state
$\left\vert \psi_{v}\right\rangle $ so that the fate of the dynamics will be
either (\ref{sa-1}) or (\ref{sa-2-2}) depending on whether the URTM halts on
the input $v$ or not:

\begin{lemma}
\label{lem:main}Suppose $d:=\dim$ $\mathcal{H}\geq d_{0}$, and let
$\{\left\vert e_{\kappa}\right\rangle \}_{\kappa=0,1,2}$ be an orthonormal set
of vectors. Given also is a URTM $M$ that takes a bit string $v$ as an input.
Then there is a shift-invariant nearest neighbor Hamiltonian $H$, and a map
$f_{\psi}$ of a bit string $v$ to a state vector $\left\vert \psi\right\rangle
$ with
\begin{align}
\left\langle e_{0}\right\vert \left.  \psi\right\rangle  &  =0,\nonumber\\
\left\Vert \left\vert \psi\right\rangle \left\langle \psi\right\vert
-\left\vert e_{1}\right\rangle \left\langle e_{1}\right\vert \right\Vert _{1}
&  \leq\varepsilon_{1}, \label{Ve=e}%
\end{align}
such that the dynamics starting from (\ref{rho-L}) satisfies (\ref{sa-1}) if
$M$ halts on $v$ and $n:=\left\vert v\right\vert \geq n_{0}$. Otherwise, it
satisfies (\ref{sa-2-2}) with $\eta=1/2$. Moreover, the computation time of
the function $f_{\psi}$\ with the error at most $\varepsilon$ is in polynomial
of $n$ and $\log(1/\varepsilon)$.
\end{lemma}

\begin{lemma}
\label{lem:main-2}In the setting of the previous lemma, replace a URTM $M$ by
an RTM $M$ solving an \textsf{EXPSPACE}-complete problem. Then there is a
shift-invariant nearest neighbor Hamiltonian $H$, and a map $f_{\psi}$ of a
bit string $v$ to a state vector $\left\vert \psi\right\rangle $ with
(\ref{Ve=e}) such that the dynamics starting from (\ref{rho-L-iid}) satisfies
(\ref{sa-1}) if $M$ accepts $v$ and $n:=\left\vert v\right\vert \geq n_{0}$.
Otherwise, it satisfies (\ref{sa-2-2}) with $\eta=1/2$. Moreover, the
computation time of the function $f_{\psi}$\ with the error at most
$\varepsilon$ is in polynomial of $n$ and $\log(1/\varepsilon)$.
\end{lemma}

The proof of these main technical lemmas will be done from the next section.
The rest of the section will be devoted to the proof of the theorems from the lemmas.

\noindent\textbf{Proof of Theorem\thinspace\ref{thm:SA}, (i). }Suppose
$\mathsf{SAS}(d,H,f_{\psi},\frac{1}{2},\varepsilon_{1})$ ($\varepsilon
_{1}<\frac{1}{4}$) is decidable. Then by Lemma\thinspace\ref{lem:main},
$\mathsf{Halt}$ can be solved by reducing to it, which is contradiction. Here,
note that $\mathsf{Halt}$ is still undecidable even if the input length is
restricted to $n:=\left\vert v\right\vert \geq n_{0}$. Therefore,
$\mathsf{SAS}(d,H,f_{\psi},\frac{1}{2},\varepsilon_{1})$ is undecidable. The
proof of the other statement is almost analogous.$\blacksquare$

\noindent\textbf{Proof of Theorem\thinspace\ref{thm:SA}, (ii).} We only prove
that $\mathsf{SAH}(d,f_{H},\frac{1}{2},\varepsilon_{1})$ ($\varepsilon
_{1}<\frac{1}{4}$) is undecidable, since the other statement can be proved
almost analogously. By Lemma\thinspace\ref{lem:main}, there is a Hamiltonian
$H^{\prime}$ and $\left\vert \psi^{\prime}\right\rangle $ such that the
dynamics starting from
\[
\rho^{\prime L}:=\left\vert e_{0}\right\rangle \left\langle e_{0}\right\vert
\otimes(\left\vert \psi^{\prime}\right\rangle \left\langle \psi^{\prime
}\right\vert )^{\otimes L}%
\]
satisfies the requirements of the lemma: Here, $\overline{\rho}(t,L)$ is
replaced by $\overline{\rho}^{\prime}(t,L)$ which is defined in the same
manner as $\overline{\rho}(t,L)$, and $\varepsilon_{1}<(4+4\sqrt{2})^{-1}$.

Denote by $V_{v}$ the rotation in the two dimensional space $\,\mathrm{span}%
\,\{\left\vert \psi^{\prime}\right\rangle ,\left\vert e_{1}\right\rangle \}$
that sends $\left\vert e_{1}\right\rangle $ to $\left\vert \psi^{\prime
}\right\rangle $. ($V_{v}$ acts as the identity on its orthogonal complement
space.) Let $\left\vert \psi\right\rangle :=\left\vert e_{1}\right\rangle $,
and consider the Hamiltonian $H:=V_{v}^{\dagger}H^{\prime}V_{v}$. Then
\begin{align*}
\rho^{L}(t)  &  =e^{-\iota tH}\left\vert e_{0}\right\rangle \left\langle
e_{0}\right\vert \otimes(\left\vert e_{1}\right\rangle \left\langle
e_{1}\right\vert )^{\otimes L}e^{\iota tH}\\
&  =\left(  V_{v}^{\dagger}\right)  ^{\otimes L+1}e^{-\iota tH^{\prime}}%
(V_{v}\left\vert e_{0}\right\rangle \left\langle e_{0}\right\vert
V_{v}^{\dagger})\otimes(V\left\vert e_{1}\right\rangle \left\langle
e_{1}\right\vert V_{v}^{\dagger})^{\otimes L}e^{\iota tH^{\prime}}%
V_{v}^{\otimes L+1}\\
&  =\left(  V_{v}^{\dagger}\right)  ^{\otimes L+1}\rho^{\prime L}%
(t)V_{v}^{\otimes L+1},
\end{align*}
so
\[
\overline{\rho}(t,L)=V_{v}^{\dagger}\overline{\rho}^{\prime}(t,L)V_{v}.
\]
Therefore, for any state $\sigma$,
\begin{align*}
\left\Vert \overline{\rho}^{\prime}(t,L)-\sigma\right\Vert _{1}  &
=\left\Vert V_{v}\overline{\rho}(t,L)V_{v}^{\dagger}-\sigma\right\Vert _{1}\\
&  \leq\left\Vert \overline{\rho}(t,L)-\sigma\right\Vert _{1}+2\left\Vert
V_{v}-I\right\Vert \\
&  =\left\Vert \overline{\rho}(t,L)-\sigma\right\Vert _{1}+2\sqrt
{2(1-\left\vert \left\langle e_{1}\right.  \left\vert \psi_{v}^{\prime
}\right\rangle \right\vert )}\\
&  \leq\left\Vert \overline{\rho}(t,L)-\sigma\right\Vert _{1}+\sqrt
{2}\varepsilon_{1}%
\end{align*}
So if the $\mathsf{SAH}(d,f_{H},\frac{1}{2},\varepsilon_{1})$ with
$\varepsilon_{1}$ $<(4+4\sqrt{2})^{-1}$ can be solved by a TM, then the TM can
solve the $\mathsf{SAS}(d,H,f_{\psi},\frac{1}{2},(1+\sqrt{2})\varepsilon_{1}%
)$. Therefore, the $\mathsf{SAH}(d,f_{H},\frac{1}{2},\varepsilon_{1})$ is
undecidable if $\varepsilon_{1}$ $<(4+4\sqrt{2})^{-1}$. If $\varepsilon_{1}$
is larger, the problem can become only harder. Therefore, the $\mathsf{SAH}%
(d,f_{H},\frac{1}{2},\varepsilon_{1})$ with $\varepsilon_{1}<1/4$ is
undecidable.$\blacksquare$

\section{Sketch of the argument}

As we had explained, the proof of the main lemmas rests on the emulation of an
URTM $M$ on an input $v$. To complete with difficulties mentioned previously,
we correspond the Hamiltonian dynamics to the RTM $M_{A}$ that uses $M$ as a subroutine.

If the initial state is (\ref{rho-L-0}), each site, except a single site that
corresponds to the finite control, corresponds to a tape cell. So for the
halting of the machine $M$ to become apparent in the space average of
single-site states, majority of the tape cells of $M_{A}$ are rewritten after
the simulation of $M$ terminates. Moreover, the result of this rewriting of
cells should be almost independent of the input $v$ to $M$: recall the
asymptotic state of the Hamiltonian dynamics is irrelevant to the input $v$
except that it is halting or not. Meantime, since $M_{A}$ is reversible, the
information about $v$ cannot be erased. Therefore, $M_{A}$ is equipped with
two kinds of the cells, $M$- and $A$-cells. The input is written in the former
in a encoded manner, and the simulation of $M$ takes place only using
$M$-cells. The majority of cells are $A$-cells, which will be rewritten in the
final amplification stage. Since $M$-cells are minority, the information about
input (other than halting/non-halting) nor the process of computation of $M$
affect only negligibly the space average of the single-site states.

The input $v$ is encoded as the relative frequency of $0$ and $1$'s of a bit
string. So $\left\vert \psi\right\rangle $ is in superposition of an $A$-cell
and an $M$-cell, and the latter's register storing the input bit string is in
superposition of $\left\vert 0\right\rangle $ and $\left\vert 1\right\rangle
$. The amplitudes are defined so that the typical part of $\left\vert
\psi\right\rangle ^{\otimes L}$ corresponds to an intended initial
configuration when $L$ is very large.

Later we prove that the interference between the states corresponding to
distinct initial configurations does not affect the space average of single
site states. Here we give a rough explanation. Recall the information about
the input other than halting/non-halting, the spatial distribution of
$M$-cells, and computational process do not affect the asymptotic space
average of single site states. This means the latter is not significantly
affected by taking `partial trace` over the former. Therefore, interference
terms between these configurations are not relevant.

The case where the initial state is (\ref{rho-L-iid-0}) is treated in a
similar manner. However, in this case, $\left\vert \psi\right\rangle $ is in
superposition of an initialized finite control $\left\vert e_{0}\right\rangle
$ and tape cell. So we use sites sandwiched by $\left\vert e_{0}\right\rangle
$'s as an RTM with the finite length tape.

If the initial state is (\ref{rho-L-0}), The initial state

the signal indicating the halt of the URTM $M$ appears only in its finite
control and tape cells are not affected. So in the space average, it is
negligible. This can be amended by composing another Turing machine that
rewrites tape cells upon the halt of $M$.

Second, recall the single site observable $A$ is almost arbitrary. $A$ should
be affected only by the halting/non-halting of $M$, but it may be largely
affected by unpredictable change occurred in the course of the simulation. We
avoid this by using two types of cells, $A$-cells and $M$-cells. The former
and the latter is exclusively used in the last amplification stage and
simulation of $M$, respectively. Also, we let $M$-cells be overwhelming
minority, so the change in the $M$-cells, or equivalently the dynamics during
the simulation of $M$, has only negligible effect in the space average.

Third, the initial state have to be more or less symmetric ((\ref{rho-L-0}) or
(\ref{rho-L-iid-0})), while having two kinds of cells contradict with this
condition. Also, the input bit string to the URTM $M$ has to be encoded to the
initial state. But if it is encoded crudely (encoding $01001...$ to
$\left\vert 0\right\rangle \left\vert 1\right\rangle \left\vert 0\right\rangle
\left\vert 0\right\rangle \left\vert 1\right\rangle $, for example), the
initial state will be highly asymmetric. To circumvent the difficulty, let
$\left\vert \psi\right\rangle $ be product of a superposition of various
classical configurations, so that typical part of $\left\vert \psi
\right\rangle ^{\otimes L}$ corresponds to a `good' initial configuration.

Fifth, the theory in \cite{NagajWojcan} states the dynamics starting from a
state corresponding to a single classical configuration. So if we apply it to
the initial state (\ref{rho-L-0}) or (\ref{rho-L-iid-0}), interference between
terms corresponding to different initial classical configurations appears. But
we show that these terms have only minor effect on the space average of the
single-site state. Roughly, it is because the space average is almost
irrelevant to the information of initial configuration other than
halting/non-halting and the rate of $A$-cells. Recall $\left\vert
\psi\right\rangle $ is fabricated so that typical part of $\left\vert
\psi\right\rangle ^{\otimes L}$ is entirely hating/non-halting and has the
constant rate of $A$-cells.

\section{The RTM $M_{A}$ for the first main lemma}

We use a minor modification of quadruple form of RTM, instead of commonly used
quintuple form \cite{Morita}, for the reason that will be clear in considering
quantum analogue. The operation of the RTM $M_{A}$ has the four stages:
initialization of the tape, decoding of the input, simulation of $M$ and
amplification of the signal.

\subsection{The tape cells and their initial configurations}

The tape of $M_{A}$ has two tracks: the first track is read-only, and the
input $v$ to the URTM $M$ is written here in an encoded form. The initial
configuration of the second track does not vary with $v$, and only the second
track is rewritten.

There are two kinds of cells, $M$- cells and $A$-\thinspace cells. We simulate
the URTM $M$ (almost) only $M$- cells, and upon halting, $A$-cells are
rewritten. The relative frequency of $M$-cells is $\alpha$, which is
negligibly small:%
\[
0\leq\alpha<<1.
\]
Therefore, unless the URTM $M$ halts, most of the cells are not rewritten.

The first register of an $M$-cell is filled with a pair of bits $b=(b_{1}%
,b_{2})$, to which the input is encoded, so the set of the symbols for this
register is $\Gamma_{1,M}:=\{b=(b_{1},b_{2});b_{\kappa}=0,1\}$. Its second
register is used for the decode of the input and the simulation of $M$. The
set of symbols of this register is denoted by $\Gamma_{2,M}$. In particular
\[
\square,s_{0}\in\Gamma_{2,M},
\]
where the former and the latter indicates the left-end of the tape and a
blank, respectively.

The first register of an $A$-cell (the part of an $A$-cell corresponding to
the first track) is filled with $\varsigma_{A}$, so the set of the symbols for
this register is a singleton $\Gamma_{1,A}:=\{\varsigma_{A}\}$. Its second
register, used for the amplification of the halting signal, is one of
\[
\Gamma_{2,A}:=\{a_{\kappa},\kappa=1,\cdots,\left\vert \Gamma_{2,A}\right\vert
-1,\square\}.
\]
As repeatedly explained, an $A$-cells is used only in the amplification stage,
except it happens to be at the left end of the tape. In such a case, $A$-cell
also functions as an $M$-cell, using the symbol $\square$ that indicates the
left end of the tape. For a while, $\left\vert \Gamma_{2,A}\right\vert =2+1$:
in the amplification stage, we simply flip $a_{1}$ to $a_{2}$.

So the set of the symbols is
\[
\Gamma=(\Gamma_{1,A}\times\Gamma_{2,A})\cup(\Gamma_{1,M}\times\Gamma_{2,M}).
\]
Initially, each $A$-cell and $M$-cell is in $(b,s_{0})$ or $(\varsigma
_{A},a_{1})$, respectively.

\subsection{Finite control}

The finite control of the RTM $M_{A}$ has the following structure:%

\begin{align*}
Q  &  :=Q_{m}\times Q_{u},\\
Q_{m}  &  :=\{m_{0},m_{1}\}.
\end{align*}
$Q_{m}$ is used to distinguish the read-write-mode ($m_{0}$) and shift-mode
($m_{1}$): In the former, the tape head read and write the tape cell which is
pointing at and $Q_{u}$, and changes the $Q_{m}$ from $m_{0}$ to $m_{1}$. In
the latter, the tape head is shifted depending solely on the state of $Q_{u}$
of the finite control, without reading any cell: it shifts to the right by one
step iff $q\in Q_{u,+}$. to the right iff $q\in Q_{u,-}$, does not shift iff
$q\in Q_{u,0}$. Any two of $Q_{u,+}$, $Q_{u,-}$, and $Q_{u,0}$ are disjoint
(This corresponds to unique direction property.).

$Q_{u}$ is divided into the four subsets%
\[
Q_{u}=Q_{u_{1}}\cup Q_{u_{2}}\cup Q_{u_{3}}\cup Q_{u_{4}}.
\]
Each of them corresponds to one of the stages of the operations (the
initialization of the tape, the decoding of the input, the simulation of $M$
and the amplification of the signal).

The Turing machine $M_{A}$ is reversible iff the move function of $M_{A}$ is
invertible. To show this, let us construct the TM $M^{\prime}$ that invert the
move of $M_{A}$: If the $Q_{m}$-register is $m_{1}$, $M^{\prime}$ does
read-write operation according to the inverse of the read-write rule of
$M_{A}$, then changes its $Q_{m}\,$-\thinspace register to $m_{0}$. If the
$Q_{m}\,$-\thinspace register is $m_{0}$, the tape-head is shifted according
to the inverse of the rule of $M_{A}$ (the left shift iff $q\in Q_{u,+}$,
etc.) and its $Q_{m}\,$-\thinspace part is changed to $m_{1}$.

\subsection{The move of the RTM $M_{A}$}

\subsubsection{Initialization of the tape cells}

To regulate the move of $M_{A}$, the left-most cell is indicated by the symbol
$\square$. However, we cannot write it at the left end in the initial
configuration because of the form (\ref{rho-L}). So at the first step, the
tape head, which is at the left-most cell, rewrite $(\varsigma_{A}%
,a_{1})\rightarrow(\varsigma_{A},\square)$ ($A$-cell) or $(b,s_{0}%
)\rightarrow(\varsigma_{A},\square)$ ($M$-cell).

The initial finite control state is $(m_{0},q_{1,init})$, and then on
rewriting of the cell, it turns to $(m_{1},q_{2,init})$. Then the tape head
moves to the right, changing the finite control state to $(m_{0}%
,q_{2,init})\in Q_{u_{2}}$.

\subsubsection{Encode of the input and its decode}

The input string will be encoded in the following manner. Below, $b_{i_{k}}$
($k=1,2,\cdots$) indicates that it is the input register of the $k$-th
$M$-cell after the cell with $\square$.

The input bit string is encoded to the rate of $1$'s in the sequence of the
first bits $b_{1,i_{1}},b_{1,i_{2}},\cdots$:The fractional part of the binary
expansion of the relative frequency equals the input string. Here, to make the
end of the bit string explicit, we suppose the input string always end with
$1$.

In decoding, we clearly need an upperbound $n^{\prime}$ to the input length
$n$, that is encoded to the sequence of the second bits $b_{2,i_{1}%
},b_{2,i_{2}},\cdots$ \ in the following manner:%
\begin{equation}
b_{2,i_{k}}=\left\{
\begin{array}
[c]{cc}%
0, & k=1,\cdots,n^{\prime}\\
1, & k=n^{\prime}+1\\
\text{arbitrary} & \text{otherwise}.
\end{array}
\right.  \label{b2}%
\end{equation}
For example, $n^{\prime}=3$, $b_{2,i_{1}}b_{2,i_{2}}b_{2,i_{3}}b_{2,i_{4}}=$
$0001$. This $n^{\prime}$ is used to define the number of the $b_{1,i}$'s that
is used to compute the relative frequency $1$'s. We use $2^{4n^{\prime}}$of them.

Denote the relative frequency of $1$'s in $2^{4n^{\prime}}$of $b_{1,i}$'s by
$\beta^{\prime}$, and let $\beta:=0.v_{1}v_{2}\cdots v_{n}$, where
$v=(v_{1},v_{2},\cdots,v_{n})$ is the input string. $\beta^{\prime}$ is
contained in the open interval $(\beta-2^{-(n^{\prime}+1)},\beta
+2^{-(n^{\prime}+1)})$: If $\beta^{\prime}\geq\beta$, clearly they
considerable to the $n^{\prime}+1$ places. Next suppose $\beta^{\prime}<\beta
$. Then they coincide in their first $n-1$ places. The $n$-th place of
$\beta^{\prime}$ is $0$, and its digits from the $n+1$-th to the $n^{\prime
}+1$-th places are $1$. In either case, we can recover $\beta$ from
$\beta^{\prime}$.

\textit{In decoding, we do not rewrite the input registers: this is important
to kill unwanted effect of the interference which otherwise may occur in the
corresponding quantum system.} When the decoding finishes, the tape head is
brought to the left-most cell, and upon reading $\square$, the finite control
state turns to $\left(  m_{1},q_{3,init}\right)  $. Then the tape head shifts
to the right, and the finite control state becomes $\left(  m_{0}%
,q_{3,init}\right)  $. From here, the simulation of the URTM $M$ starts.

\subsubsection{The simulation of the URTM $M$}

Here, we essentially uses the move relation of the URTM $M$, except that all
the $A$-cells, except for the one at the left end (if any), are skipped.
$M_{A}$ is reversible if and only if $M$ does not go into a loop for the
correctly formatted initial condition.

There are many such URTM: For example, if $M^{\prime}$ is a universal TM, by
adding another track on the tape and leaving the record of the all the move of
the TM $M^{\prime}$ to that additional track, we can easily construct a URTM
that does not fall into the loop. Also, the URTM that emulate the cyclic tag
system, which is another universal computational model, does not go into a
loop, as the tag keep moving to the right at the end of every cycle of the move.

During the simulation, we only uses $M$-cells' third register, and the other
registers of the cells are not modified.

Upon halting,or the state becomes $(m_{1},q_{3,halt})$, \ the tape head is
moved to the left, until it reads $\square$. Upon reading the symbol, the
finite control state becomes $(m_{1},q_{4,init})$, while changing the
$A$-register of the left most cell from $a_{1}$ to $a_{2}$. Then the
amplification stage starts.

When $M$ does not halt, this stage continues forever.

\subsubsection{The amplification of the halting signal}

In this last stage, it simply flips the $A$-cells from $a_{1}$ to $a_{2}$,
while the tape head keep moving to the left: the finite control state is
either $(m_{1},q_{4,init})$ or $(m_{0},q_{4,init})$: In case of the former,
the tape head moves to the right, and the finite control state becomes
$(m_{0},q_{3,halt})$ . In the case of the latter, the $A$-cell, except for the
one marked by $\square$, is flipped from $a_{1}$ to $a_{2}$, and the finite
control state becomes $(m_{1},q_{4,init})$.

Since the tape is half-infinite, this stage continues forever.

\section{The quantum system for the first main lemma}

We corresponds the restriction of $M_{A}$ to the tape with the length $L$ to a
quantum system with $L+1$-sites. Each site $\mathcal{H}$is spanned by a CONS
$\{\left\vert x\right\rangle ;x\in X\}$, where the set $X$ is the union of the
set of the tape symbols $\Gamma$ and the finite control states $Q$,
\[
X=Q\cup\Gamma.
\]
We use notations such as%
\begin{align*}
\mathcal{H}  &  =\mathcal{H}^{Q}\oplus\mathcal{H}^{\Gamma},\\
\mathcal{H}^{Q}  &  =\mathcal{H}^{Q_{m}}\otimes\mathcal{H}^{Q_{u}},\\
\mathcal{H}^{\Gamma}  &  =(\mathcal{H}^{\Gamma_{1,M}}\otimes\mathcal{H}%
^{\Gamma_{2,M}})\oplus(\mathcal{H}^{\Gamma_{1,A}}\otimes\mathcal{H}%
^{\Gamma_{2,A}}).
\end{align*}
The classical configuration of the machine is represented by $\boldsymbol{x}%
=(x_{0},x_{1},\cdots x_{L})$. Whenever it is necessary, we use symbol such as
$\boldsymbol{x}^{L}$, $M_{A}^{L}$, etc. to indicate the lattice size. Only one
of $x_{i}$'s is an element of $Q$, and others are elements of $\Gamma$. If
$x_{i_{0}}\in Q$, the tape head is pointing to the $i_{0}+1$-th cell,
represented by $x_{i_{0}+1}$.

$\boldsymbol{x}^{L}$ is a \textit{legal initial configuration} if it
corresponds a initial configuration of the RTM $M_{A}$ that emulates $M$ for
an input bit string: Here, $x_{L+1},x_{L+2},\cdots$ are arbitrarily fixed. So
$L$ should be large enough for an $\boldsymbol{x}^{L}$ to be a legal initial
configuration. $x_{0}\in Q$ for all the legal initial configurations. Also,
$\boldsymbol{x}^{L}$ is a \textit{legal configuration} if it corresponds to a
configuration of the restriction of $M_{A}$ to the tape with $L$-cells
starting from a legal initial configuration.

\subsection{The Hamiltonian $H^{L}$}

The Hamiltonian $H^{L}$ is the sum of the isometry $U^{L}$ and its dual
$(U^{L})^{\dagger}$,
\[
H^{L}=U^{L}+(U^{L})^{\dagger}.
\]
From here, we drop the superscript $L$, and write $H^{L}$ and $U^{L}$ simply
by $H$ and $U$, unless they are too confusing.

Here $U$ sends a classical state $\left\vert \boldsymbol{x}\right\rangle
=\otimes_{i=0}^{L}\left\vert x_{i}\right\rangle $ to another such state
$\left\vert \boldsymbol{x}^{\prime}\right\rangle =\otimes_{i=0}^{L}\left\vert
x_{i}^{\prime}\right\rangle $,
\[
U\otimes_{i=0}^{L}\left\vert x_{i}\right\rangle =\otimes_{i=0}^{L}\left\vert
x_{i}^{\prime}\right\rangle .
\]
Moreover, if $\left\vert \boldsymbol{x}\right\rangle :=$\ $\otimes_{i=0}%
^{L}\left\vert x_{i}\right\rangle $ corresponds to a legitimate initial
configuration of $M_{A}$,
\[
\left\vert j;\boldsymbol{x}\right\rangle :=(U)^{j}\left\vert \boldsymbol{x}%
\right\rangle
\]
corresponds to the configuration at the $j$-th step.

Such $U$ with only local terms is composed in the following manner. $U$ is
decomposed into the sum $U:=\sum_{i=-L}^{L}U_{i}$, where $U_{i}$ acts
nontrivially only on the $i-1$-th, $i$-th and $i+1$-th sites, and
$U_{i}\left\vert \boldsymbol{x}\right\rangle $ is non-zero only if the $i$-th
site corresponds to the finite control.

All $U_{i}$'s are identical in the case of the periodic boundary condition.
Each $U_{i}$ is further decomposed into the sum of 1- and 2- body terms:
\[
U_{i}:=U_{i}^{0}+U_{i}^{1+}+U_{i}^{1-}+U_{i}^{10}\text{ }(i=1,\cdots,L-1).
\]
Here, $U_{i}^{0}$ implements the move of the $M_{A}$ whence the $Q_{m}%
$\thinspace-\thinspace register of the finite control is $0$. As we had
supposed the tape head is reading the right neighbor of the finite control
site, it acts on the $i$-th and the $i+1$-th site:%
\[
U_{i}^{0}:(\mathbb{C}\left\vert m_{0}\right\rangle \otimes\mathcal{H}^{Q_{u}%
})_{i}\otimes\mathcal{H}_{i+1}^{\Gamma}\mapsto(\mathbb{C}\left\vert
m_{1}\right\rangle \otimes\mathcal{H}^{Q_{u}})_{i}\otimes\mathcal{H}%
_{i+1}^{\Gamma}.
\]

Meantime, \ $U_{i}^{1-}$ corresponds to the right shift, so acts on the
$(i-1)$ -th and the $i$-th site:
\[
U_{i}^{1-}:\mathcal{H}_{i-1}^{\Gamma}\otimes(\mathbb{C}\left\vert
m_{1}\right\rangle \otimes\mathcal{H}^{Q_{u,-}})_{i}\mapsto(\mathbb{C}%
\left\vert m_{0}\right\rangle \otimes\mathcal{H}^{Q_{u,-}})_{i-1}%
\otimes\mathcal{H}_{i}^{\Gamma},
\]
where the content of $\mathcal{H}^{Q_{u,-}}$ and $\mathcal{H}^{\Gamma}$ are
not changed. $U_{i}^{1+}$, $U_{i}^{10}$ are defined analogously.

\begin{remark}
If the move of the machine is not divided into two modes $m_{0}$ and $m_{1}$,
a three body term is necessary to implement the leftward shift.
\end{remark}

No illegal configuration should not possess its successor. In particular, if
the tape is half-infinite, the head can never read the $\square$-ed cell after
having shifted rightward. So no successor of such an illegal configuration
should be defined:%

\[
U_{i}^{0}\left\vert m_{0},q\right\rangle _{i}\left\vert c,\square\right\rangle
_{i+1}=0,\,\forall q\in Q_{u,+}\,,c\in\Gamma_{1,A}\cup\Gamma_{1,M}.
\]
Here recall other interaction terms are also null,
\[
U_{i}^{1x}\left\vert m_{0},q\right\rangle _{i}\left\vert c,\square
\right\rangle _{i+1}=0,\,(x=\pm,0).
\]
Also, we drop any interaction between $\left\vert m_{\kappa^{\prime}%
},q\right\rangle _{i}$ ($q\in Q_{u,+}$) and $\left\vert c,\square\right\rangle
_{i+1}$, so that the interaction between $L+1$-th and the $1$-st site is
effectively cut off.

In case of the periodic boundary condition, such a configuration occurs when
the finite control comes to the $L+1$-th site by shifting to the right, and no
successor is not defined. If this occurs at the $J_{\boldsymbol{x}}$-th step,%
\[
U^{J_{\boldsymbol{x}}+1}\left\vert \boldsymbol{x}\right\rangle =U\left\vert
J_{\boldsymbol{x}};\boldsymbol{x}\right\rangle =0.
\]

\subsection{The correspondence between $M_{A}$ and the Hamiltonian dynamics}

\label{subsec:dynamics}

The continuous time dynamics $e^{-\iota tH}\left\vert \boldsymbol{x}%
\right\rangle $ starting from $\left\vert \boldsymbol{x}\right\rangle $ is not
a continuous interpolation of $\{\left\vert j;\boldsymbol{x}\right\rangle
;j=1,\cdots,J_{\boldsymbol{x}}\}$ nor approximation to it. But it had been
known that they are related with each other in the following way
\cite{NagajWojcan}. Here, be careful not to confuse two different "time" :
$j=1,\cdots J_{\boldsymbol{x}}$, that characterize the move of $M_{A}$, and
$t$, that characterize the dynamics of the Hamiltonian cell automata
corresponding to $M_{A}$.

First, $e^{-\iota tH}\left\vert \boldsymbol{x}\right\rangle $ does not go out
of the span of $\{\left\vert j;\boldsymbol{x}\right\rangle ;j=1,\cdots
,J_{\boldsymbol{x}}\}$,
\[
e^{-\iota tH}\left\vert \boldsymbol{x}\right\rangle \in\mathrm{span}%
\{\left\vert j;\boldsymbol{x}\right\rangle ;j=1,\cdots,J_{\boldsymbol{x}}\}
\]
and
\begin{align}
e^{-\iota Ht}\left\vert \boldsymbol{x}\right\rangle  &  =\frac{2}%
{J_{\boldsymbol{x}}+1}\sum_{k=1}^{J_{\boldsymbol{x}}}\sum_{j=1}%
^{J_{\boldsymbol{x}}}e^{\iota\omega_{k,\boldsymbol{x}}t}\sin\frac{\pi
k}{J_{\boldsymbol{x}}+1}\sin\frac{jk\pi}{J_{\boldsymbol{x}}+1}\left\vert
j;\boldsymbol{x}\right\rangle ,\nonumber\\
\omega_{k,\boldsymbol{x}}  &  :=2\cos\frac{2\pi k}{J_{\boldsymbol{x}}+1}.
\label{hca-t}%
\end{align}
(See, e.g., \cite{NagajWojcan}).

Second, if the system start from a state corresponding to a classical
configuration and $T$ is very large, all the steps are visited almost
uniformly: Here, we defined the `time-averaged probability'
$p_{j;\boldsymbol{x}}$ by \
\begin{equation}
p_{j;\boldsymbol{x}}:=\lim_{T\rightarrow\infty}\int_{0}^{T}\left\vert
\left\langle j;\boldsymbol{x}\right\vert e^{-\iota tH}\left\vert
\boldsymbol{x}\right\rangle \right\vert ^{2}dt=\left\{
\begin{array}
[c]{cc}%
\frac{1}{J_{\boldsymbol{x}}+1}, & j=2,\cdots,J_{\boldsymbol{x}}-1,\\
\frac{3}{2}\frac{1}{J_{\boldsymbol{x}}+1}, & j=1,J_{\boldsymbol{x}}.
\end{array}
\right.  \label{pj}%
\end{equation}
Observe we had defined the move of the RTM $M_{A}$ so that $J_{\boldsymbol{x}%
}\rightarrow\infty$ as $L\rightarrow\infty$. So $p_{j;\boldsymbol{x}}$ is
almost uniform when $L$ is very large.

\subsubsection{Comparison with other uses of Feyman-Kitaev Hamiltonian}

This $H$ is essentially a Feyman-Kitaev Hamiltonian used in the various
subjects in quantum computation, such as the proof of the undecidability of
the spectral gap. This study bears some similarity to ours, but they use the
ground state and ground energy, while we use the dynamics starting from a
legitimate initial state. So there are some differences in the construction of
the Hamiltonians. As a while, ours is much easier.

First, they implement clock counters, so that $\left\vert j;\boldsymbol{x}%
\right\rangle $ contains a register indicating $j$ so that $\left\langle
j;\boldsymbol{x}\right.  \left\vert j^{\prime};\boldsymbol{x}\right\rangle =0$
holds for all $j\neq j^{\prime}$. In our case, however, the RTM $M_{A}$ is
classical and deterministic. So no clock-counter is necessary.

Second, in our case, there is no need to handle with illegal configurations
where $M_{A}$ moves in a unexpected manner: we simply use initial state that
has no or negligible overlap with those `bad' configurations. If the ground
state of the Hamiltonian matters, this is not the case: Illegal configuration
should be carefully penalized so that they won't contribute to the ground
state. (Note the penalty term has to be implemented by local interactions!)

So as a whole, our composition of the Hamiltonian is much simpler than theirs.

\subsection{The initial quantum state}

The 1-D lattice with the size $L+1$ corresponds to the $M_{A}$ with the tape
length $L$. In the initial setting, the $0$-th site stands for the finite
control, and the sites in its right stand for the cells of the tape. So in the
case of the open boundary condition, only the right half are used.

Define the state vector $\left\vert e_{0}\right\rangle $ corresponding to a
initial configuration of the finite control:
\[
\left\vert e_{0}\right\rangle :=\left\vert m_{0},q_{1,init}\right\rangle
\in\mathcal{H}^{Q}.\
\]
Define also a state vector representing a configuration of a cell
\[
\left\vert e_{\kappa}\right\rangle :=\left\vert \varsigma_{A},a_{\kappa
}\right\rangle \,\in\mathcal{H}^{\Gamma}\ (\kappa\geq1).
\]
Clearly, $\{\left\vert e_{\kappa}\right\rangle ;\kappa=0,1,\cdots
,|\Gamma_{2,A}|\}$ are orthogonal.

If $n=|v|\leq n_{0}$, we define
\[
\left\vert \psi\right\rangle :=\left\vert \varsigma_{A}\right\rangle
\left\vert a_{1}\right\rangle =\left\vert e_{1}\right\rangle .
\]

Meantime, if $n=|v|>n_{0}$, we define $\left\vert \psi\right\rangle $ by\
\begin{align*}
\left\vert \psi\right\rangle  &  :=\sqrt{1-\alpha}\left\vert \varsigma
_{A}\right\rangle \left\vert a_{1}\right\rangle +\sqrt{\alpha}\left\vert
\mathrm{input}\right\rangle \left\vert s_{0}\right\rangle \\
&  =\sqrt{1-\alpha}\left\vert e_{1}\right\rangle +\sqrt{\alpha}\left\vert
\mathrm{input}\right\rangle \left\vert s_{0}\right\rangle ,
\end{align*}
where the input $v$ is encoded to $\left\vert \mathrm{input}\right\rangle
\in\,\mathrm{span}\,\{\left\vert b\right\rangle ;b\in\Gamma_{1,M}\}$:
\begin{equation}
\left\vert \mathrm{input}\right\rangle =(\sqrt{1-\beta}\left\vert
0\right\rangle +\sqrt{\beta}\left\vert 1\right\rangle )\otimes(\sqrt{1-n^{-2}%
}\left\vert 0\right\rangle +\sqrt{n^{-2}}\left\vert 1\right\rangle ),
\label{input}%
\end{equation}
where $0\leq\beta<1$, $n$ is the length of the input bit string. The
fractional part of the binary representation of $\beta$ equals the input bit
string, with the promise that the last bit equals $1$. When $L$ is large, the
decode succeeds for overwhelming portion of the initial configurations.

In the initial state $\left\vert e_{0}\right\rangle (\left\vert \psi
_{v}\right\rangle )^{\otimes L}$, most of the sites are in superposition of
$A$-cells and $M$-cells, while the former shares overwhelmingly large amplitude.

\begin{remark}
The classical picture of the dynamics bears some stochastic aspect, but the
randomness exists only in the initial configuration, and the process is
completely deterministic.
\end{remark}

Here we let
\begin{equation}
\alpha:=(\varepsilon_{1}/4)^{2}\leq1/(4\cdot4)^{2}. \label{alpha}%
\end{equation}
It is easy to check the conditions (\ref{Ve=e}).

\section{Proof of the first main lemma}

\subsection{"Dephasing" between the initial configurations\ }

Our initial state is in the super position of various classical (and legal)
configurations, while the correspondence between computational process and the
Hamiltonian dynamics was discussed for the case where the initial state is in
a single classical configuration. However, as we demonstrate below, we can
safely replace the initial state by the probabilistic mixture of classical configurations.

Since we are interested in the distance between $\overline{\rho}(t,L)$ and
$\left\vert e_{1}\right\rangle \left\langle e_{1}\right\vert $, $\frac{1}%
{2}(\left\vert e_{1}\right\rangle \left\langle e_{1}\right\vert +\left\vert
e_{2}\right\rangle \left\langle e_{2}\right\vert )$, it suffices to compute
\[
\left\langle e_{\kappa^{\prime}}\right\vert \overline{\rho}(t,L)\left\vert
e_{\kappa}\right\rangle =\,\mathrm{tr}\,\rho^{L}(t)\left(  \left\vert
e_{\kappa}\right\rangle \left\langle e_{\kappa^{\prime}}\right\vert \right)
^{(L)},\,
\]
where $\left(  \left\vert e_{\kappa}\right\rangle \left\langle e_{\kappa
^{\prime}}\right\vert \right)  ^{(L)}$ is defined in analogy with $A^{(L)}$.
Observe the observable $B=\left\vert e_{\kappa}\right\rangle \left\langle
e_{\kappa^{\prime}}\right\vert $ satisfies
\begin{equation}
B\left\vert x\right\rangle =0,x\in Q\text{ or }\Gamma_{1,M}\times\Gamma_{2,M}.
\label{Bx=0}%
\end{equation}
In computing the expectation of the space average $B^{(L)}$ such an observable
$B$, we argue that the superposition of the legal configuration can be treated
as a probabilistic mixture:

\begin{lemma}
\label{lem:cofig-deco}Let $\boldsymbol{y}$ and $\boldsymbol{y}^{\prime}$ be
arbitrary configurations. Suppose $B$ satisfies (\ref{Bx=0}). Then
$\left\langle \boldsymbol{y}^{\prime}\right\vert B^{(L)}\,\left\vert
\boldsymbol{y}\right\rangle \not =0$ only if all the following holds:

(i) For any $i_{0}$ with $y_{i_{0}}\in Q$, it holds that $y_{i_{0}}^{\prime
}\in Q$ and $y_{i_{0}}=y_{i_{0}}^{\prime}$.

(ii) For all $i_{0}$ with $y_{i_{0}}\in\Gamma_{1,M}\times\Gamma_{2,M}$, \ it
holds that $y_{i_{0}}^{\prime}\in\Gamma_{1,M}\times\Gamma_{2,M}$ and
$y_{i_{0}}=y_{i_{0}}^{\prime}$.
\end{lemma}

\begin{proof}
Suppose $y_{i_{0}}\in Q$ and $y_{i_{0}}^{\prime}\neq y_{i_{0}}$. Then
$\left\langle \boldsymbol{y}^{\prime}\right\vert B_{i_{0}}\left\vert
\boldsymbol{y}\right\rangle =0$ by $B_{i_{0}}\left\vert \boldsymbol{y}%
\right\rangle =0$. If $i\neq i_{0}$, $B_{i}$ acts trivially on $\mathcal{H}%
_{i_{0}}$, and $\left\langle \boldsymbol{y}^{\prime}\right\vert B_{i}%
\left\vert \boldsymbol{y}\right\rangle =0$. Therefore, $y_{i_{0}}^{\prime
}=y_{i_{0}}$, so the condition (i) should be satisfied. The argument for the
condition (ii) is almost parallel.
\end{proof}

\begin{lemma}
\label{lem:init-deco}Let $\boldsymbol{x}$ and $\boldsymbol{x}^{\prime}$ be a
legal initial configurations. Suppose $B$ is an observable with (\ref{Bx=0}).
Then if $\boldsymbol{x\neq x}^{\prime}$,%
\[
\,\left\langle \boldsymbol{x}^{\prime}\right\vert e^{\iota tH}B^{(L)}\text{
}e^{-\iota tH}\left\vert \boldsymbol{x}\right\rangle =0.
\]

\end{lemma}

\begin{proof}
As $H=U+U^{\dagger}$, it suffices to prove
\[
\left\langle \boldsymbol{x}^{\prime}\right\vert (U^{\dagger})^{j^{\prime}%
-1}B^{(L)}U^{j-1}\left\vert \boldsymbol{x}\right\rangle =\left\langle
j^{\prime};\boldsymbol{x}^{\prime}\right\vert B^{(L)}\left\vert
j;\boldsymbol{x}\right\rangle =0
\]
for all $j$ and $j^{\prime}$. If $\left\vert j;\boldsymbol{x}\right\rangle $
and $\left\vert j^{\prime};\boldsymbol{x}^{\prime}\right\rangle $ differ in
the position of the finite control site, then the identity is true by the
previous lemma. So suppose they coincide in the position of the finite
control. Then if the $i$-th site of $\left\vert j;\boldsymbol{x}\right\rangle
$ corresponds to the $i^{\prime}$-th cell ($i^{\prime}=i$ or $i-1$), so does
the $i$-th site of of $\left\vert j^{\prime};\boldsymbol{x}^{\prime
}\right\rangle $. As $\boldsymbol{x}\neq\boldsymbol{x}^{\prime}$ and the first
registers are read-only, $\left\vert j;\boldsymbol{x}\right\rangle $ and
$\left\vert j^{\prime};\boldsymbol{x}^{\prime}\right\rangle $ differs either
in the position or the content of $M$-cells. Therefore, by the previous lemma,
$\left\langle j^{\prime};\boldsymbol{x}^{\prime}\right\vert B^{(L)}\left\vert
j;\boldsymbol{x}\right\rangle =0$.
\end{proof}

Define $\rho_{\boldsymbol{x}}^{L}(t):=$ $e^{\iota tH}\left\vert \boldsymbol{x}%
\right\rangle \left\langle \boldsymbol{x}\right\vert e^{-\iota tH}$,
\ $\rho_{\boldsymbol{x},i}^{L}(t):=$ $\left.  e^{\iota tH}\left\vert
\boldsymbol{x}\right\rangle \left\langle \boldsymbol{x}\right\vert e^{-\iota
tH}\right\vert _{\mathcal{H}_{i}}$, and $\overline{\rho}_{\boldsymbol{x}%
}(t,L):=\frac{1}{L+1}\sum_{i=0}^{L}\rho_{\boldsymbol{x},i}^{L}$. If the
initial state is $\rho^{L}=\sum_{\boldsymbol{x},\boldsymbol{x}^{\prime}}%
\xi_{\boldsymbol{x}}\overline{\xi_{\boldsymbol{x}^{\prime}}}\left\vert
\boldsymbol{x}\right\rangle \left\langle \boldsymbol{x}^{\prime}\right\vert $,
since $B=\left\vert e_{k}\right\rangle \left\langle e_{\kappa^{\prime}%
}\right\vert $ satisfies the hypothesis of the lemmas,
\begin{align}
\,\left\langle e_{\kappa^{\prime}}\right\vert \overline{\rho}(t,L)\left\vert
e_{k}\right\rangle  &  =\,\,\mathrm{tr}\,\left(  \left\vert e_{k}\right\rangle
\left\langle e_{\kappa^{\prime}}\right\vert \right)  ^{(L)}\rho^{L}%
(t)\nonumber\\
&  =\sum_{\boldsymbol{x},\boldsymbol{x}^{\prime}}\xi_{\boldsymbol{x}}%
\overline{\xi_{\boldsymbol{x}^{\prime}}}\left\langle \boldsymbol{x}^{\prime
}\right\vert e^{\iota tH}\left(  \left\vert e_{k}\right\rangle \left\langle
e_{\kappa^{\prime}}\right\vert \right)  ^{(L)}e^{-\iota tH}\left\vert
\boldsymbol{x}\right\rangle \nonumber\\
&  =\sum_{\boldsymbol{x}}|\xi_{\boldsymbol{x}}|^{2}\left\langle \boldsymbol{x}%
\right\vert e^{\iota tH}\left(  \left\vert e_{k}\right\rangle \left\langle
e_{\kappa^{\prime}}\right\vert \right)  ^{(L)}e^{-\iota tH}\left\vert
\boldsymbol{x}\right\rangle \nonumber\\
&  =\sum_{\boldsymbol{x}}|\xi_{\boldsymbol{x}}|^{2}\,\left\langle
e_{\kappa^{\prime}}\right\vert \overline{\rho}_{\boldsymbol{x}}(t,L)\left\vert
e_{k}\right\rangle . \label{dephase-init}%
\end{align}
Therefore, we can safely replace $\rho^{L}$ with $\sum_{\boldsymbol{x}%
}\left\vert \xi_{\boldsymbol{x}}\right\vert ^{2}\left\vert \boldsymbol{x}%
\right\rangle \left\langle \boldsymbol{x}\right\vert $.

\subsection{`Good' initial configurations}

Most of the classical configurations comprising the`effective' initial state
$\sum_{\boldsymbol{x}}\left\vert \xi_{\boldsymbol{x}}\right\vert
^{2}\left\vert \boldsymbol{x}\right\rangle \left\langle \boldsymbol{x}%
\right\vert $ are `good' in the following sense:

\begin{description}
\item[(G-a)] The rate of $M$-cells falls into the interval $(\alpha
-L^{-1/3},\alpha+L^{-1/3}).$

\item[(G-b)] The decoding operation ends correctly, reading at most
$2^{4n^{3}}$ of $M$-cells.
\end{description}

Denote by $P_{\mathrm{good}}$ the projection onto the good configurations, and
denote by $P_{\mathrm{good}}^{(a)}$ $\ $and $P_{\mathrm{good}}^{(b)}$the
projection onto the configuration satisfying (G-a) and (G-b), respectively. As
we show soon,
\begin{align}
\,\mathrm{tr}\,P_{\mathrm{good}}^{(a)}\,\rho^{L}  &  \geq1-\frac{1}{4}%
L^{-1/3},\label{p-good-1}\\
\mathrm{tr}\,P_{\mathrm{good}}^{(b)}\,\rho^{L}  &  \geq1-\frac{7}{4}n^{-1},
\label{p-good-2}%
\end{align}
where the latter is true if \
\begin{equation}
n\geq n_{0}:=4\alpha^{-1}\geq16\varepsilon_{1}^{-1}\geq64. \label{n>8}%
\end{equation}
Below, we also suppose
\begin{equation}
L\geq L_{0}:=2n^{3}. \label{L>n}%
\end{equation}
\ 

As $\,P_{\mathrm{good}}^{(a)}$ and $P_{\mathrm{good}}^{(b)}$ commute,
\begin{align}
1-\,\mathrm{tr}\,P_{\mathrm{good}}\,\rho^{L}  &  \leq(1-\mathrm{r}%
\,P_{\mathrm{good}}^{(a)}\,\rho^{L})+(1-\mathrm{r}\,P_{\mathrm{good}}%
^{(b)}\,\rho^{L})\nonumber\\
&  \leq\frac{1}{4}L^{-1/3}+\frac{7}{4}n^{-1}\nonumber\\
&  \leq2n^{-1}=:p_{e}(n). \label{p-good}%
\end{align}

Recall that not all the initial classical configuration bears expected
properties. So the analysis such as (\ref{pj}) so on applies only to those
initial configuration in the support of $P_{\mathrm{good}}$. We define
\[
\rho_{\mathrm{good}}^{L}:=\sum_{\boldsymbol{x:}\text{`good'}}\left\vert
\xi_{\boldsymbol{x}}\right\vert ^{2}\left\vert \boldsymbol{x}\right\rangle
\left\langle \boldsymbol{x}\right\vert =P_{\mathrm{good}}\,\rho
P_{\mathrm{good}},
\]
and define $\rho_{\mathrm{good}}^{L}(t)$, $\rho_{\mathrm{good},i}^{L}(t)$, and
$\,\overline{\rho}_{\mathrm{good}}(t,L)$ in analogy with $\rho^{L}(t)$,
$\rho_{i}(t)$, and $\,\overline{\rho}(t,L)$, respectively. They closely
approximate $\rho^{L}$, $\rho^{L}(t)$, $\rho_{i}(t)$,$\,$\ and $\overline
{\rho}(t,L)$.

\subsubsection{Evaluation of the `bad' rate}

(\ref{p-good-1}) is easily clear by the Chernoff bound. In the sequel we
derive (\ref{p-good-2}).

In the decoding operation, use many copies of $\left\vert \mathrm{input}%
\right\rangle $, where $n$ is the length of the input to $M$. First, decode
$n^{\prime}$, which is an upper bound to $n$, using the relation (\ref{b2})
from $b_{2,i}$'s : then as the bit $1$ occurs with the probability $n^{-2}$,
\begin{align*}
\Pr\left\{  n\leq n^{\prime}\leq n^{3}\right\}   &  \geq\Pr\left\{  n^{\prime
}\geq n\right\}  -\Pr\left\{  n^{\prime}>n^{3}\right\} \\
&  =(1-n^{-2})^{n-1}-(1-n^{-2})^{n^{3}}\\
&  \geq1-n^{-1}-(1-n^{-2})^{n^{3}}\\
&  \geq1-n^{-1}-e^{-n}%
\end{align*}
Here we had used the relations $(1-c)^{d}\geq1-cd$ and $(1-x^{-1})^{x}\leq
e^{-1}$ $(x\geq1)$. So with high probability, $n^{\prime}$ is an upper bound
to $n$.

Second, use $2^{4n^{\prime}}$ ($\leq2^{4n^{3}}$) of $b_{\iota,1}$'s\ to
evaluate $\beta$. Then, the decode will succeed if $\beta^{\prime}$, the
relative frequency of $1$'s, differs from $\beta$ at most by $\beta
2^{-(n^{\prime}+1)}$. \ The probability of success is evaluated by Chernoff
bound :
\begin{align*}
\Pr\left\{  \left\vert \beta^{\prime}-\beta\right\vert \leq\beta
2^{-(n^{\prime}+2)}\right\}   &  \geq1-2\exp(-\frac{1}{3}\beta2^{-2(n^{\prime
}+1)}\,2^{4n^{\prime}})\\
&  =1-2\exp(-\frac{1}{3}\beta2^{2n^{\prime}-2})\\
&  \geq1-2\exp(-\frac{1}{3}2^{n-2})\,
\end{align*}
where the last inequality holds since the fractional part of $\beta$ ends at
the $n$-th place and $n^{\prime}\geq n$.

Therefore, we have
\[
\mathrm{tr}\,P_{\mathrm{good}}^{(b)}\,\rho^{L}\geq1-n^{-1}-e^{-n}-2\exp
(-\frac{1}{3}2^{n-2})
\]
that leads to (\ref{p-good-2}) if (\ref{n>8}).

\subsection{Approximation of $\overline{\rho}(t,L)$}

Since we are interested in the distance between $\overline{\rho}(t,L)$ and a
state supported on $\mathrm{span}\,\{\left\vert e_{\kappa}\right\rangle
\}_{\kappa\geq1}$, approximate $\overline{\rho}(t,L)$ by $P_{E}\overline{\rho
}(t,L)P_{E}$, where $P_{E}$ is the projection onto the span of
$\,\mathrm{span}\,\{\left\vert e_{\kappa}\right\rangle \}_{\kappa\geq1}$. The
error of the approximation is
\begin{align*}
&  \left\Vert P_{E}\overline{\rho}(t,L)P_{E}-\overline{\rho}(t,L)\right\Vert
_{1}\leq2\sqrt{\,1-\mathrm{tr}\,\overline{\rho}(t,L)P_{E}}\\
&  =2\sqrt{\,1-\sum_{\boldsymbol{x}}|\xi_{\boldsymbol{x}}|^{2}\mathrm{tr}%
\,\overline{\rho}_{\boldsymbol{x}}(t,L)P_{E}}\leq2\sum_{\boldsymbol{x}}%
|\xi_{\boldsymbol{x}}|^{2}\sqrt{\,1-\mathrm{tr}\,\overline{\rho}%
_{\boldsymbol{x}}(t,L)P_{E}}\\
&  \leq2\sum_{\boldsymbol{x}\text{:'good`}}|\xi_{\boldsymbol{x}}|^{2}%
\sqrt{\,1-\mathrm{tr}\,\overline{\rho}_{\boldsymbol{x}}(t,L)P_{E}}%
+2\cdot2n^{-1}.
\end{align*}
where the equality in the second line is by (\ref{dephase-init}).

Define $P_{A}:=\left\vert \varsigma_{A}\right\rangle \left\langle
\varsigma_{A}\right\vert \otimes I_{2}$. Then if $\boldsymbol{x}$ is legal
configuration,
\[
\left\langle \boldsymbol{x}\right\vert (P_{E}^{(L)}-P_{A}^{(L)})\left\vert
\boldsymbol{x}\right\rangle =\left\langle \boldsymbol{x}\right\vert \left(
\left\vert \varsigma_{A},\square\right\rangle \left\langle \varsigma
_{A},\square\right\vert \right)  ^{(L)}\left\vert \boldsymbol{x}\right\rangle
\leq\frac{1}{L+1}\text{.}%
\]
Also, as the number of the $A$-cells is constant of the dynamics,
\[
\left\langle \boldsymbol{x}\right\vert e^{itH}P_{A}^{(L)}e^{-itH}\left\vert
\boldsymbol{x}\right\rangle =\left\langle \boldsymbol{x}\right\vert
e^{itH}P_{A}^{(L)}e^{-itH}\left\vert \boldsymbol{x}\right\rangle .
\]
Therefore
\begin{align*}
\mathrm{tr}\,\overline{\rho}_{\boldsymbol{x}}(t,L)P_{E}  &  \geq
\mathrm{tr}\,\overline{\rho}_{\boldsymbol{x}}(t,L)P_{A}-\frac{1}{L+1}\\
&  =\mathrm{tr}\,\overline{\rho}_{\boldsymbol{x}}(0,L)P_{A}-\frac{1}{L+1}%
\end{align*}
so
\begin{align*}
&  \left\Vert P_{E}\overline{\rho}(t,L)P_{E}-\overline{\rho}(t,L)\right\Vert
_{1}\\
&  \leq2\sum_{\boldsymbol{x}\text{:'good`}}|\xi_{\boldsymbol{x}}|^{2}%
\sqrt{\,1-\mathrm{tr}\,\overline{\rho}_{\boldsymbol{x}}(0,L)P_{A}+\frac
{1}{L+1}}+4n^{-1}\\
&  \leq2\sum_{\boldsymbol{x}\text{:'good`}}|\xi_{\boldsymbol{x}}|^{2}%
\sqrt{1-\frac{L}{L+1}(1-\,\alpha-L^{-1/3})+\frac{1}{L+1}}+4n^{-1}\\
&  \leq2\sqrt{\alpha+\frac{2}{L}+L^{-1/3}}+4n^{-1}.
\end{align*}

Also,
\begin{align*}
P_{E}\overline{\rho}(t,L)P_{E}  &  =\,\sum_{\kappa,\kappa^{\prime}\geq
1}\left\langle e_{\kappa}\right\vert \overline{\rho}(t,L)\left\vert
e_{\kappa^{\prime}}\right\rangle \left\vert e_{\kappa}\right\rangle
\left\langle e_{\kappa^{\prime}}\right\vert \\
&  =\,\sum_{\boldsymbol{x}}|\xi_{\boldsymbol{x}}|^{2}\sum_{\kappa
,\kappa^{\prime}\geq1}\left\langle e_{\kappa}\right\vert \overline{\rho
}_{\boldsymbol{x}}(t,L)\left\vert e_{\kappa^{\prime}}\right\rangle \left\vert
e_{\kappa}\right\rangle \left\langle e_{\kappa^{\prime}}\right\vert \\
&  =\,\sum_{\boldsymbol{x}}|\xi_{\boldsymbol{x}}|^{2}P_{E}\overline{\rho
}_{\boldsymbol{x}}(t,L)P_{E},
\end{align*}
where the equality in the second line is by (\ref{dephase-init}). So if
$\rho_{\ast}$ is an arbitrary state,%
\begin{align}
&  \left\Vert \overline{\rho}(t,L)-\rho_{\ast}\right\Vert _{1}\nonumber\\
&  \leq\left\Vert P_{E}\overline{\rho}(t,L)P_{E}-\overline{\rho}%
(t,L)\right\Vert _{1}+\left\Vert P_{E}\overline{\rho}(t,L)P_{E}-\rho_{\ast
}\right\Vert _{1}\nonumber\\
&  \leq\left\Vert P_{E}\overline{\rho}(t,L)P_{E}-\overline{\rho}%
(t,L)\right\Vert _{1}+\sum_{\boldsymbol{x}}|\xi_{\boldsymbol{x}}%
|^{2}\left\Vert P_{E}\overline{\rho}_{\boldsymbol{x}}(t,L)P_{E}-\rho_{\ast
}\right\Vert _{1}\nonumber\\
&  \leq\left\Vert P_{E}\overline{\rho}(t,L)P_{E}-\overline{\rho}%
(t,L)\right\Vert _{1}+\max_{\boldsymbol{x}\text{:'good`}}\left\Vert
P_{E}\overline{\rho}_{\boldsymbol{x}}(t,L)P_{E}-\rho_{\ast}\right\Vert
_{1}+2\cdot2n^{-1}\nonumber\\
&  \leq2\sqrt{\alpha+\frac{2}{L}+L^{-1/3}}+8n^{-1}+\max_{\boldsymbol{x}%
\text{:'good`}}\left\Vert P_{E}\overline{\rho}_{\boldsymbol{x}}(t,L)P_{E}%
-\rho_{\ast}\right\Vert _{1}. \label{error-rho-e-1-1}%
\end{align}

Analogously,
\begin{align}
\left\Vert \lim_{T\rightarrow\infty}\frac{1}{T}\int_{0}^{T}dt\,\overline{\rho
}(t,L)-\rho_{\ast}\right\Vert _{1}  &  \leq2\sqrt{\alpha+\frac{2}{L}+L^{-1/3}%
}+8n^{-1}\nonumber\\
&  +\max_{\boldsymbol{x}\text{:'good`}}\left\Vert \lim_{T\rightarrow\infty
}\frac{1}{T}\int_{0}^{T}dt\,P_{E}\overline{\rho}_{\boldsymbol{x}}%
(t,L)P_{E}-\rho_{\ast}\right\Vert _{1}. \label{error-rho-e-1-2}%
\end{align}
(Here, $\boldsymbol{x}$ runs over a finite set, $\sum_{\boldsymbol{x}}$ can be
exchanged with $\lim_{T\rightarrow\infty}\frac{1}{T}\int_{0}^{T}dt$.)

\subsection{"Dephasing" between the time steps}

Here we show
\[
\lim_{T\rightarrow\infty}\frac{1}{T}\int_{0}^{T}dt\,\mathrm{tr}\,\overline
{\rho}_{\boldsymbol{x}}(t,L)B\approx\sum_{j=1}^{J_{\boldsymbol{x}}%
}p_{j;\boldsymbol{x}}\left\langle j;\boldsymbol{x}\right\vert B^{(L)}%
\left\vert j;\boldsymbol{x}\right\rangle ,
\]
where $p_{j;\boldsymbol{x}}$ is as of (\ref{pj}) and the error of the
approximation is
\begin{equation}
\left\vert \lim_{T\rightarrow\infty}\frac{1}{T}\int_{0}^{T}dt\,\mathrm{tr}%
\,\overline{\rho}_{\boldsymbol{x}}(t,L)B-\sum_{j=1}^{J_{\boldsymbol{x}}%
}p_{j;\boldsymbol{x}}\left\langle j;\boldsymbol{x}\right\vert B^{(L)}%
\left\vert j;\boldsymbol{x}\right\rangle \right\vert \leq\frac{2}{L}\left\Vert
B\right\Vert . \label{deco-step}%
\end{equation}
By approximating $p_{j;\boldsymbol{x}}$ by $1/J_{\boldsymbol{x}}$,
\begin{equation}
\left\vert \lim_{T\rightarrow\infty}\frac{1}{T}\int_{0}^{T}dt\,\mathrm{tr}%
\,\overline{\rho}_{\boldsymbol{x}}(t,L)B-\frac{1}{J_{\boldsymbol{x}}}%
\sum_{j=1}^{J_{\boldsymbol{x}}}\left\langle j;\boldsymbol{x}\right\vert
B^{(L)}\left\vert j;\boldsymbol{x}\right\rangle \right\vert \leq\left\{
\frac{2}{L}+\frac{2}{J}\right\}  \left\Vert B\right\Vert . \label{deco-step-2}%
\end{equation}
The results in \cite{NagajWojcan} justifies the relation only for diagonal
observables, so we have to give the proof slightly generalizing their analysis.

In the rest of the section, we show (\ref{deco-step}). By (\ref{hca-t}),%

\begin{align*}
&  \lim_{T\rightarrow\infty}\frac{1}{T}\int_{0}^{T}dt\,\mathrm{tr}%
\,\overline{\rho}_{\boldsymbol{x}}(t,L)B=\lim_{T\rightarrow\infty}\frac{1}%
{T}\int_{t=0}^{T}dt\left\langle \boldsymbol{x}^{\prime}\right\vert
e^{-itH}B^{(L)}e^{\iota tH}\left\vert \boldsymbol{x}\right\rangle \\
&  =\frac{4}{(J_{\boldsymbol{x}}+1)^{2}}\sum_{k,k^{\prime}=1}%
^{J_{\boldsymbol{x}}}\sum_{j,j^{\prime}=1}^{J_{\boldsymbol{x}}}\left[
\lim_{T\rightarrow\infty}\frac{1}{T}\int_{t=0}^{T}e^{\iota\left(
\omega_{k,\boldsymbol{x}}-\omega_{k,\boldsymbol{x}}\right)  t}dt\right] \\
&  \quad\quad\quad\quad\quad\quad\quad\quad\quad\quad\quad\quad\times\sin
\frac{\pi k^{\prime}}{J_{\boldsymbol{x}}+1}\sin\frac{\pi k}{J_{\boldsymbol{x}%
}+1}\sin\frac{j^{\prime}k^{\prime}\pi}{J_{\boldsymbol{x}}+1}\sin\frac{jk\pi
}{J_{\boldsymbol{x}}+1}\left\langle j^{\prime};\boldsymbol{x}^{\prime
}\right\vert B^{(L)}\left\vert j;\boldsymbol{x}\right\rangle \\
&  =\frac{4}{(J_{\boldsymbol{x}}+1)^{2}}\sum_{k=1}^{J_{\boldsymbol{x}}}%
\sum_{j,j^{\prime}=1}^{J_{\boldsymbol{x}}}\sin^{2}\frac{\pi k}%
{J_{\boldsymbol{x}}+1}\sin\frac{j^{\prime}k\pi}{J_{\boldsymbol{x}}+1}\sin
\frac{jk\pi}{J_{\boldsymbol{x}}+1}\left\langle j^{\prime};\boldsymbol{x}%
^{\prime}\right\vert B^{(L)}\left\vert j;\boldsymbol{x}\right\rangle \\
&  =\sum_{j=1}^{J_{\boldsymbol{x}}}p_{j;\boldsymbol{x}}\left\langle
j;\boldsymbol{x}\right\vert B^{(L)}\left\vert j;\boldsymbol{x}\right\rangle
-\frac{1}{2}\frac{1}{J_{\boldsymbol{x}}+1}\sum_{1\leq j,,j^{\prime}\leq
J_{\boldsymbol{x}},j^{\prime}=j\pm2}\left\langle j;\boldsymbol{x}\right\vert
B^{(L)}\left\vert j^{\prime};\boldsymbol{x}\right\rangle ,
\end{align*}
where the last identity is by
\begin{align*}
&  \sum_{k=1}^{J}\sin^{2}\frac{\pi k}{J+1}\sin\frac{j^{\prime}k\pi}{J+1}%
\sin\frac{jk\pi}{J+1}\\
&  =\left\{
\begin{array}
[c]{cc}%
\frac{1}{4}(J+1), & j=j^{\prime}\neq1\text{ and}\neq J,\\
\frac{1}{8}(3J+3), & j=j^{\prime}=1\text{ or}=J.\\
0, & j\neq j^{\prime}\ \text{and }j^{\prime}\neq j\pm2,\\
-\frac{1}{8}(J+1), & j\neq j^{\prime}\ \text{and }j^{\prime}\neq j\pm2.
\end{array}
\right.
\end{align*}

We argue the second term of the last end almost vanishes. Observe
\[
\left\vert \left\langle j;\boldsymbol{x}\right\vert B^{(L)}\left\vert
j^{\prime};\boldsymbol{x}\right\rangle \right\vert \leq\frac{1}{L+1}\sum
_{i=0}^{L}\left\vert \left\langle j;\boldsymbol{x}\right\vert B_{i}\left\vert
j^{\prime};\boldsymbol{x}\right\rangle \right\vert
\]
The configurations corresponding to $\left\vert j^{\prime};\boldsymbol{x}%
\right\rangle $ and $\left\vert j;\boldsymbol{x}\right\rangle $ differ at the
$i_{0}$-th site, for example. Then, $\left\langle j;\boldsymbol{x}\right\vert
B_{i}\left\vert j^{\prime};\boldsymbol{x}\right\rangle =0$ unless $i=i_{0}$.
Therefore,
\[
\sum_{i=0}^{L}\left\vert \left\langle j;\boldsymbol{x}\right\vert
B_{i}\left\vert j^{\prime};\boldsymbol{x}\right\rangle \right\vert
\leq\left\vert \left\langle j;\boldsymbol{x}\right\vert B_{i_{0}}\left\vert
j^{\prime};\boldsymbol{x}\right\rangle \right\vert \leq\left\Vert B\right\Vert
.
\]
Therefore,%

\begin{align*}
\frac{1}{J_{\boldsymbol{x}}+1}|\sum_{1\leq j,,j^{\prime}\leq J_{\boldsymbol{x}%
},j^{\prime}=j\pm2}\left\langle j;\boldsymbol{x}\right\vert B^{(L)}\left\vert
j^{\prime};\boldsymbol{x}\right\rangle |  &  \leq\frac{1}{J_{\boldsymbol{x}%
}+1}\sum_{1\leq j,,j^{\prime}\leq J_{\boldsymbol{x}},j^{\prime}=j\pm2}%
\frac{\left\Vert B\right\Vert }{L+1}\\
&  \leq\frac{2\left\Vert B\right\Vert }{L+1}\leq\frac{2\left\Vert B\right\Vert
}{L},
\end{align*}
leading to the asserted relation.

\subsection{Energy gaps}

Denote by $P_{\boldsymbol{x}}$ the projector onto the subspace spanned by
$\{\left\vert j;\boldsymbol{x}\right\rangle ;j=1,\cdots,J_{\boldsymbol{x}}\}$:
Here, $\boldsymbol{x}$ is not necessarily a legal configuration, so the the
dynamics $U^{j}\left\vert \boldsymbol{x}\right\rangle $ $(j\in\mathbb{N})$ may
be cyclic. In such a case, $J_{\boldsymbol{x}}$ denotes the period.

Clearly, $\,P_{\boldsymbol{x}}P_{\boldsymbol{x}^{\prime}}=0$ ($\boldsymbol{x}%
\neq\boldsymbol{x}^{\prime}$) and
\[
\mathcal{H}^{\otimes L+1}=\,\oplus_{\boldsymbol{x}\text{ }}\mathrm{supp}%
\,P_{\boldsymbol{x}}\text{,}%
\]
and $P_{\boldsymbol{x}}HP_{\boldsymbol{x}^{\prime}}$ vanishes if
$\boldsymbol{x}^{\prime}\neq\boldsymbol{x}$. If the dynamics starting from
$\left\vert \boldsymbol{x}\right\rangle $ is not cyclic, the eigenvalues of
$P_{\boldsymbol{x}}HP_{\boldsymbol{x}}$ is $\omega_{k,x}$ as of (\ref{hca-t}).
If it is cyclic, the eigenvalues are $\cos\frac{2\pi k}{J_{\boldsymbol{x}}}$
($k=0,1,...$), and the corresponding eigenvector(s) is (are):
\[
\sum_{j=1}^{J_{\boldsymbol{x}}}e^{\iota\frac{2\pi k}{J_{\boldsymbol{x}}}%
(j-1)}\left\vert j;\boldsymbol{x}\right\rangle ,\,\sum_{j=1}%
^{J_{\boldsymbol{x}}}e^{-\iota\frac{2\pi k}{J_{\boldsymbol{x}}}(j-1)}%
\left\vert j;\boldsymbol{x}\right\rangle .
\]
(For some values of $J_{\boldsymbol{x}}$ and $k$, these two are identical
modulo constant factor.)

So without loss of generality, \ any two eigenvalues $\omega$ and
$\omega^{\prime}$ are in the following form:%
\[
\omega=\cos\frac{2\pi k}{m},\omega^{\prime}=\cos\frac{2\pi k^{\prime}%
}{m^{\prime}},
\]
where
\[
\max\{m,m^{\prime}\}\leq L
\]

\bigskip

By (\ref{hca-t}), the difference between two energy levels is bounded below as follows.%

\begin{align*}
\left\vert 2\cos\frac{2\pi k}{J_{\boldsymbol{x}}+1}-2\cos\frac{2\pi\left(
k+1\right)  }{J_{\boldsymbol{x}}+1}\right\vert  &  \geq\left\vert 2\cos
\frac{2\pi J_{\boldsymbol{x}}}{J_{\boldsymbol{x}}+1}-2\cos\frac{2\pi\left(
J_{\boldsymbol{x}}+1\right)  }{J_{\boldsymbol{x}}+1}\right\vert \\
&  =4\sin^{2}\frac{\pi}{J_{\boldsymbol{x}}+1}\\
&  \geq4\left(  \frac{1}{\pi/2}\frac{\pi}{J_{\boldsymbol{x}}+1}\right)  ^{2}\\
&  =\frac{8}{(J_{\boldsymbol{x}}+1)^{2}}\\
&  \geq\frac{8}{(d^{L+1}+1)^{2}}.
\end{align*}
Here, the last inequality is

\subsection{Proof}

\subsubsection{The case where the input is too short}

If the input length $n=|v|\leq n_{0}$, we had defined $V_{v}=I$. So there is
no $M$-cells in its initial configuration, and $U$ simply shifts the finite
control cell rightward, without modifying anything else. Therefore, this case
reduces to the non-halting case with $\alpha=1$ and $\rho_{\mathrm{good}}%
^{L}=\rho^{L}$.

\subsubsection{The case where the URTM $M$ does not halt}

Suppose $L$ is so small that the decoding or the simulation does not
terminate. Then the amplification stage does not start. Therefore, we are in
trouble only if there are enough $M$-cells and but the configuration is `bad'.
So suppose the initial configuration $\boldsymbol{x}$ is `good'. In this case,
$M_{A}$ does not proceed to the amplification stage, so all the $A$-cells,
except perhaps the one at the left end, are in $a_{1}$ forever. So if
$\boldsymbol{x}$ is a `good' initial configuration,%
\[
P_{E}^{(L)}\left\vert j;\boldsymbol{x}\right\rangle =\left(  \left\vert
e_{1}\right\rangle \left\langle e_{1}\right\vert \right)  ^{(L)}\left\vert
j;\boldsymbol{x}\right\rangle ,
\]
so%

\begin{align}
\left\Vert P_{E}\overline{\rho}_{\boldsymbol{x}}(t,L)P_{E}-\left\vert
e_{1}\right\rangle \left\langle e_{1}\right\vert \right\Vert _{1}  &
=\left\Vert \left\langle e_{1}\right\vert \overline{\rho}_{\boldsymbol{x}%
}(t,L)\left\vert e_{1}\right\rangle \left\vert e_{1}\right\rangle \left\langle
e_{1}\right\vert -\left\vert e_{1}\right\rangle \left\langle e_{1}\right\vert
\right\Vert _{1}\nonumber\\
&  =1-\left\langle e_{1}\right\vert \overline{\rho}_{\boldsymbol{x}%
}(t,L)\left\vert e_{1}\right\rangle =1-\,\mathrm{tr}\,P_{E}\,\overline{\rho
}_{\boldsymbol{x}}(t,L)\nonumber\\
&  \leq\alpha+\frac{2}{L}+L^{-1/3}, \label{rho-e1}%
\end{align}
where we had used (\ref{PE-tr-2}) to obtain the last inequality. Therefore, by
(\ref{error-rho-e-1-1}), (\ref{sa-2-2}) is satisfied:%

\begin{align*}
\left\Vert \overline{\rho}(t,L)-\left\vert e_{1}\right\rangle \left\langle
e_{1}\right\vert \right\Vert _{1}  &  \leq2\sqrt{\alpha+\frac{2}{L}+L^{-1/3}%
}+4n^{-1}+\alpha+\frac{2}{L}+L^{-1/3}\\
&  \leq4\sqrt{\alpha}=\varepsilon_{1},
\end{align*}
where the second inequality is by (\ref{n>8}), (\ref{L>n}), and (\ref{alpha}).

\subsubsection{The case where the URTM $M$ halts}

In this case, we use (\ref{deco-step-2}). First, let $B=\left\vert
e_{\kappa^{\prime}}\right\rangle \left\langle e_{\kappa}\right\vert $ and
$\kappa\neq\kappa^{\prime}$, $\left\langle j;\boldsymbol{x}\right\vert
B_{i}\left\vert j;\boldsymbol{x}\right\rangle =0$ for all $i$, so
\[
\frac{1}{J_{\boldsymbol{x}}}\sum_{j=1}^{J_{\boldsymbol{x}}}\left\langle
j;\boldsymbol{x}\right\vert (\left\vert e_{\kappa^{\prime}}\right\rangle
\left\langle e_{\kappa}\right\vert )^{(L)}\left\vert j;\boldsymbol{x}%
\right\rangle =0.
\]

Next, consider $\left\langle j;\boldsymbol{x}\right\vert (\left\vert
e_{2}\right\rangle \left\langle e_{2}\right\vert )^{(L)}\left\vert
j;\boldsymbol{x}\right\rangle $, which equals $N_{2}(j)/(L+1)$, where
$N_{2}(j)$ is the number of $a_{2}$ at the $j$-th step.

Suppose the third stage ends at the $j_{0}$-th step, and the number of
$M$-cells not marked by $\square$ is $L_{m}$. If $\boldsymbol{x}$ is a `good'
initial configuration,
\[
(\alpha-L^{-1/3})L<L_{m}<(\alpha+L^{-1/3})L.
\]

As the head sweeps all the cells from the left end to the right end,
\[
J_{\boldsymbol{x}}=j_{0}+2L.
\]

Observe $\frac{1}{J_{x}}\sum_{j=1}^{J_{x}}N_{2}(j)$ is the area below the
graph $j\rightarrow N_{2}(j)$. This is complicated function of the
distribution of $M$-cells, and it takes minimum if all the $M$-cells are
clustered in the left end. In this case,
\begin{align*}
\frac{1}{J_{x}(L+1)}\sum_{j=1}^{J_{x}}N_{2}(j)  &  =\frac{1}{J_{x}(L+1)}%
\sum_{k=1}^{L-L_{m}-2}2k\text{ }\\
&  =\frac{(L-L_{m}-2)(L-L_{m}-1)}{(2L+j_{0})(L+1)}\\
&  \rightarrow\frac{1}{2}-\alpha,\text{ as }L\rightarrow\infty.
\end{align*}
On the other hand, $\frac{1}{J_{x}}\sum_{j=1}^{J_{x}}N_{2}(j)$ is maximized if
all the $M$-cells are clustered in the right end. In this case, after $N_{2}$
increased to $L-L_{m}-1$, the head sweeps the $M$-cells while $N_{2}$ kept
unchanged. Therefore,
\begin{align*}
\frac{1}{J_{x}(L+1)}\sum_{j=1}^{J_{x}}N  &  =\frac{1}{J_{x}(L+1)}\{\sum
_{k=1}^{L-L_{m}-2}2k+2L_{m}(L-L_{m}-1)\}\\
&  =\frac{(L+L_{m}-2)(L-L_{m}-1)}{(2L+j_{0})(L+1)}\\
&  \rightarrow\frac{1-\alpha^{2}}{2},\text{as }L\rightarrow\infty.
\end{align*}
Therefore,
\[
\lim_{L\rightarrow\infty}\left\vert \frac{1}{J_{\boldsymbol{x}}}\sum
_{j=1}^{J_{\boldsymbol{x}}}\left\langle j;\boldsymbol{x}\right\vert
(\left\vert e_{2}\right\rangle \left\langle e_{2}\right\vert )^{(L)}\left\vert
j;\boldsymbol{x}\right\rangle -\frac{1}{2}\right\vert \leq\alpha.
\]

By $N_{1}(j)+N_{2}(j)=$ $L-L_{m}-1$,
\begin{align*}
\lim_{L\rightarrow\infty}\frac{1}{J_{x}(L+1)}\sum_{j=1}^{J_{x}}N_{1}(j)  &
\geq\lim_{L\rightarrow\infty}\frac{L-L_{m}-1}{L+1}-\frac{1-\alpha^{2}}{2}\\
&  =\frac{1}{2}-\alpha+\frac{\alpha^{2}}{2},\\
\lim_{L\rightarrow\infty}\frac{1}{J_{x}(L+1)}\sum_{j=1}^{J_{x}}N_{1}(j)  &
\leq\lim_{L\rightarrow\infty}\frac{L-L_{m}-1}{L+1}-(\frac{1}{2}-\alpha
)=\frac{1}{2}.
\end{align*}
Therefore
\[
\lim_{L\rightarrow\infty}\left\vert \frac{1}{J_{\boldsymbol{x}}}\sum
_{j=1}^{J_{\boldsymbol{x}}}\left\langle j;\boldsymbol{x}\right\vert
(\left\vert e_{1}\right\rangle \left\langle e_{1}\right\vert )^{(L)}\left\vert
j;\boldsymbol{x}\right\rangle -\frac{1}{2}\right\vert \leq\alpha.
\]

So by (\ref{deco-step-2}),
\begin{align*}
\lim_{L\rightarrow\infty}\left\vert \lim_{T\rightarrow\infty}\frac{1}{T}%
\int_{0}^{T}dt\left\langle e_{\kappa}\right\vert \overline{\rho}%
_{\boldsymbol{x}}(t,L)\left\vert e_{\kappa}\right\rangle -\frac{1}%
{2}\right\vert  &  \leq\alpha+\lim_{L\rightarrow\infty}\frac{2}{L}+\frac
{2}{2L+j_{0}}\\
&  =\alpha\,\,\,\,\,\,\,(\kappa=1,2),\\
\lim_{L\rightarrow\infty}\left\vert \lim_{T\rightarrow\infty}\frac{1}{T}%
\int_{0}^{T}dt\left\langle e_{1}\right\vert \overline{\rho}_{\boldsymbol{x}%
}(t,L)\left\vert e_{2}\right\rangle \right\vert  &  \leq\lim_{L\rightarrow
\infty}\{\frac{2}{L}+\frac{2}{2L}\}=0.
\end{align*}

Therefore,
\[
\lim_{L\rightarrow\infty}\left\Vert \lim_{T\rightarrow\infty}\frac{1}{T}%
\int_{0}^{T}dtP_{E}\,\overline{\rho}_{\boldsymbol{x}}(t,L)\,P_{E}-\frac{1}%
{2}(\left\vert e_{1}\right\rangle \left\langle e_{1}\right\vert +\left\vert
e_{2}\right\rangle \left\langle e_{2}\right\vert )\right\Vert _{1}\leq
2\alpha.
\]
So by (\ref{error-rho-e-1-2}),%
\begin{align*}
&  \lim_{L\rightarrow\infty}\left\Vert \lim_{T\rightarrow\infty}\frac{1}%
{T}\int_{0}^{T}dt\overline{\rho}(t,L)\,-\frac{1}{2}(\left\vert e_{1}%
\right\rangle \left\langle e_{1}\right\vert +\left\vert e_{2}\right\rangle
\left\langle e_{2}\right\vert )\right\Vert _{1}\leq2\sqrt{\alpha}%
+4n^{-1}+2\alpha\\
&  \leq2\sqrt{\alpha}+3\alpha\leq4\sqrt{\alpha}=\varepsilon_{1}\text{ }%
\end{align*}
where we had used (\ref{n>8}) and (\ref{alpha}). Therefore,\ (\ref{sa-1}) is satisfied.

\section{Open boundary condition}

In case of the open boundary condition, the Hamiltonian $H^{L}$ for the finite
size system is defined by omitting terms involving the non-existent sites:
Then if the tape cell runs out, the dynamics is aborted.

As for the initial state, clearly, the position of the state $\left\vert
e_{0}\right\rangle $ in the 1-D lattice is important. Here, we assume
$\left\vert e_{0}\right\rangle $ is at the middle, so
\[
\rho^{L}=\left(  \left\vert \psi\right\rangle \left\langle \psi\right\vert
\right)  ^{\otimes L/2}\otimes\left\vert e_{0}\right\rangle \left\langle
e_{0}\right\vert \otimes\left(  \left\vert \psi\right\rangle \left\langle
\psi\right\vert \right)  ^{\otimes L/2}.
\]
Then in the last amplification stage, only the sites in the right half line
will be rewritten. So by the analysis analogous to the case of periodic
boundary condition, we can still prove the statement of the first main lemma,
except that $\eta=1/2$ in the statement should be replaced by $\eta=1/4$.
Accordingly, the parameters in the theorems should be modified.

Moreover, as explained below, by adopting more complicated amplification
stage, the value of the parameter $\eta$ can be set to an arbitrary value in
$(0,2/3]$.

For this purpose, we rewrite the cells at the both ends of the $0$-th one. For
this purpose, the head moves to the right and left alternatingly. Now we use
the tape symbols
\begin{align*}
\Gamma_{2,A}^{\prime}  &  :=\{a_{1},a_{2},a_{3},\square,\square_{2}%
,\square_{3}\}\\
\Gamma_{2,M}^{\prime}  &  :=\Gamma_{2,M}\cup\{\square_{2},\square_{3}\}.
\end{align*}
instead of $\Gamma_{2,A}$, $\Gamma_{2,M}$. Also,
\[
Q_{u,4}:=\{q_{4,0},q_{4,1},q_{4,2},q_{4,3}\}.
\]
\ 

At the end of the simulation stage, the tape head is moved to the $\square$-ed
cell, and upon reading the cell, the finite control state is changed from
$(m_{0},q_{3,halt})$ to $(m_{1},q_{4,1})$, while rewriting $\square$ to
$\square_{3}$, irrespective of it is an $A$-cell or not. In the succeeding
operations, the other $M$-cells are skipped. After this, the machine operates
according to the table:

$%
\begin{tabular}
[c]{l|lllll|l}
& \multicolumn{5}{|c|}{$m_{0}$} & $m_{1}$\\\cline{1-6}\cline{2-6}
& $a_{1}$ & $a_{2}$ & $a_{3}$ & $\square_{2}$ & $\square_{3}$ & \\\hline
$q_{4,0}$ & $a_{1},q_{4,0}$ & $a_{2},q_{4,0}$ & $a_{2},q_{4,1}$ & $\square
_{2},q_{4,0}$ &  & $\rightarrow$\\
$q_{4,1}$ & $a_{3},q_{4,2}$ &  &  &  &  & $\rightarrow$\\
$q_{4,2}$ & $a_{1},q_{4,2}$ & $a_{2},q_{4,2}$ & $a_{2},q_{4,3}$ & $\square
_{2},q_{4,2}$ & $\square_{2},q_{4,3}$ & $\leftarrow$\\
$q_{4,3}$ & $a_{3},q_{4,0}$ &  &  &  &  & $\leftarrow$%
\end{tabular}
\ \ \ $

If the entity of the table is empty, no successor is defined.

In the sequel, we compute the average of the rate of $a_{2}$. Let $N_{2}(j)$
be the number of $a_{2}$ at the step $j$. With the two-way move, the number of
the steps between the increment of $N_{2}(j)$ is proportional to $N_{2}(j)$.
So, $O(k)^{2}$-steps are necessary from the start of the amplification stage
to become $N_{2}(j)=k$, and
\[
N_{2}(j)=O(\sqrt{j-j_{0}}),
\]
where $j_{0}$ is the final step of the simulation stage. Therefore,
\[
\lim_{\alpha\rightarrow0}\lim_{L\rightarrow\infty}\frac{1}{J(L+1)}\sum
_{j=1}^{J}N_{2}(j)=\int_{0}^{1}\sqrt{x}\,dx=\frac{2}{3}.
\]
So the statement of Lemma\thinspace\ref{lem:main} holds if $\eta=1/2$ in the
statement is replaced by $\eta=2/3$, and the relation between $\varepsilon
_{1}$ and $\alpha$ is modified.

It is possible to decrease $\eta$ to the arbitrary value in $(0,2/3]$, by
leaving some of $A$-cells unchanged.

\section{$\mathsf{RE}$-completeness}

\label{sec:re-complete}

Though we cannot solve the halting problem, there is a TM $M^{\prime}$ that
"confirm" the answer of the problem for `Yes' instances: $M^{\prime}$ halts
and returns the answer `Yes' if the URTM $M$ halts on a input $v$. An example
of such a Turing machine operates as $M$ does on the input $v$, except it
returns `Yes' upon hatting. (The TM $M^{\prime}$ does not halt if the answer
is `No'. ) If a TM can confirm the answer of the given decision problem for
`Yes' instances and does not halt for all `No' instances, we say the problem
is recursively enumerable ($\mathsf{RE}$). The halting problem is an instance
of $\mathsf{RE}$ problems, and it is an $\mathsf{RE}$-complete problem in the
sense that any $\mathsf{RE}$ problem can be reduced to the halting problem. We
show that some versions of our problem is also $\mathsf{RE}$-complete, meaning
that the problem is exactly as difficult as the halting problem.

To prove it is in $\mathsf{RE}$, suppose the input satisfies (\ref{sa-1}) and
falsify (\ref{sa-2}).

Discretize the parameter $t$,
\begin{equation}
t_{i+1}:=t_{i}+\frac{1}{4\left\Vert H^{L}\right\Vert }(\eta-\varepsilon_{1}),
\label{dt}%
\end{equation}
so that, for any $t\in\lbrack t_{i},t_{i+1}]$%

\begin{align*}
\,\left\Vert \overline{\rho}(t_{i},L)-\overline{\rho}(t,L)\right\Vert _{1}  &
\leq2\left\Vert e^{-\iota(t-t_{i})H^{L}}-I\right\Vert \\
&  \leq2\left\Vert H^{L}\right\Vert (t-t_{i})\leq\frac{1}{2}(\eta
-\varepsilon_{1}).
\end{align*}
Here we had used $|e^{\iota x}-1|\leq\left\vert x\right\vert $. So if
$T=K\,(t_{i+1}-t_{i})$,
\begin{align*}
\left\Vert \frac{1}{T}\int_{0}^{T}dt\overline{\rho}(t,L)-\frac{1}{K}\sum
_{i=1}^{K}\overline{\rho}(t_{i},L)\right\Vert _{1}  &  \leq\frac{1}{K}%
\sum_{i=1}^{K}\frac{1}{t_{i}-t_{i-1}}\int_{t_{i-1}}^{t_{i}}dt\left\Vert
\overline{\rho}(t,L)-\overline{\rho}(t_{i},L)\right\Vert _{1}\\
&  \leq\frac{1}{2}(\eta-\varepsilon_{1}).
\end{align*}
\ 

Let $\overline{\rho}_{\mathrm{ap}}(t,L)$ be an approximation of $\,\overline
{\rho}(t,L)$ whose components are binary fractional numbers with finite digits
such that
\[
\left\Vert \overline{\rho}_{\mathrm{ap}}(t,L)-\overline{\rho}(t,L)\right\Vert
_{1}\leq\frac{1}{2}(\eta-\varepsilon_{1}).
\]
Then
\[
\left\Vert \frac{1}{T}\int_{0}^{T}dt\overline{\rho}(t,L)-\frac{1}{K}\sum
_{i=1}^{K}\overline{\rho}_{\mathrm{ap}}(t_{i},L)\right\Vert _{1}\leq
\eta-\varepsilon_{1}.
\]

So if
\begin{equation}
\,\left\Vert \frac{1}{K}\sum_{i=1}^{K}\overline{\rho}_{\mathrm{ap}}%
(t_{i},L)-\left\vert e_{1}\right\rangle \left\langle e_{1}\right\vert
\right\Vert _{1}>\varepsilon_{1}+(\eta-\varepsilon_{1})+\frac{1}{4}%
(\eta-\varepsilon_{1}) \label{check-cond-approx}%
\end{equation}
is verified by computing the norm $\left\Vert \cdot\right\Vert _{1}$ with the
error at most $\frac{1}{4}(\eta-\varepsilon_{1})$ for some $K$ and $L\geq
L_{0}$, we can conclude that
\begin{equation}
\,\left\Vert \frac{1}{T}\int_{0}^{T}dt\,\overline{\rho}(t,L)-\left\vert
e_{1}\right\rangle \left\langle e_{1}\right\vert \right\Vert _{1}%
>\varepsilon_{1}. \label{check-cond}%
\end{equation}
Therefore, (\ref{sa-2}) is falsified.

Here we show that there is $K$ and $L\geq L_{0}$ satisfying
(\ref{check-cond-approx}) whenever (\ref{sa-1}) is true. If (\ref{sa-1}) is
true, there is $L\geq L_{0}$ and a $T\geq0$ such that
\[
\left\Vert \frac{1}{T}\int_{0}^{T}dt\,\overline{\rho}(t,L)-((1-\eta)\left\vert
e_{1}\right\rangle \left\langle e_{1}\right\vert +\eta\left\vert
e_{2}\right\rangle \left\langle e_{2}\right\vert )\right\Vert _{1}%
<\varepsilon_{1}+(\eta-\varepsilon_{1}).
\]
If this condition is met,
\[
\left\Vert \frac{1}{T}\int_{0}^{T}dt\,\overline{\rho}(t,L)-\left\vert
e_{1}\right\rangle \left\langle e_{1}\right\vert \right\Vert _{1}%
\geq\varepsilon_{1}+2(\eta-\varepsilon_{1}),
\]
so
\begin{align*}
\left\Vert \frac{1}{K}\sum_{i=1}^{K}\overline{\rho}_{\mathrm{ap}}%
(t_{i},L)-\left\vert e_{1}\right\rangle \left\langle e_{1}\right\vert
\right\Vert _{1}  &  \geq\varepsilon_{1}+2(\eta-\varepsilon_{1})-\frac{1}%
{2}(\eta-\varepsilon_{1})\\
&  =\varepsilon_{1}+(\eta-\varepsilon_{1})+\frac{1}{2}(\eta-\varepsilon_{1}).
\end{align*}
This is sufficient for (\ref{check-cond-approx}).

At each $K$ and $L\geq L_{0}$, we check the condition (\ref{check-cond-approx}%
). Since the set $\{(K_{i},L);K,L\in\mathbb{N}\}$ is countable, this test can
be done for all $(K,L)$'s sequentially. The verification process terminates if
and only if the input is a `Yes' instance, so the problem is contained in
$\mathsf{RE}$.

\begin{remark}
This proof essentially demonstrates that the problem reduces to the
computation of the limit of a sequence with at most one mind-change\thinspace
\cite{DH}: It is known that the latter problem is $\mathsf{RE}$-complete
(relative to many-one reduction.).
\end{remark}

\section{The second main lemma}

Let us turn to the case where the initial state is (\ref{rho-L-iid}). The
boundary condition may be periodic or open.

In the proof, we use $v$ that corresponds to an $\mathsf{EXPSPACE}$-complete
problem : a decision problem is said to be an $\mathsf{EXPSPACE}$-complete iff
it can be solved using $\exp(poly(n))$ bits of the working space and any such
problem can be reduced to it by a polynomial time computation. For an example
of $\mathsf{EXPSPACE}$-complete problem, see

Define $\left\vert \psi\right\rangle $ as a superposition of $\left\vert
e_{0}\right\rangle $ and $V_{v}\left\vert e_{1}\right\rangle $. So,
$\left\vert \psi\right\rangle ^{\otimes L+1}$ is a superposition of classical
configurations with many $\left\vert e_{0}\right\rangle $'s, and each block of
sites between two $\left\vert e_{0}\right\rangle \,$'s are used as a simulator
of TM with a finite length.

Though the size of the system $L$ grows infinitely large irrespective of $n$,
the size of the each block is restricted by the relative frequency of
$\left\vert e_{0}\right\rangle $. If $L^{[k]}+1$ denotes the size of the
$k$-th block, we control the amplitude of $\left\vert e_{0}\right\rangle $ so
that
\[
l\lesssim L^{[k]}+1\lesssim l^{4},
\]
where $l$ is a function of $n$, and in this section we suppose
\[
l=O(\exp(poly(n)))\geq poly(n).
\]

Now the number of the steps for the first to the third sage $j_{0}$ can be as
large as $O(\exp(l))>>L^{[k]}$. So by (\ref{approx-fre}), the relative
frequency of $a_{2}$ may not be closed to $1/2$. Intuitively, this is because
$j_{0}$ is not negligible compared with the duration of the amplification
stage, which is $O(L^{[k]})=O(poly(l))$ steps.

Note it is not possible to increase $L^{[k]}$: Recall in the first and second
main technical lemma, the amplitudes of the initial configuration should be
polynomial time computable. But the amplitude for $\left\vert e_{0}%
\right\rangle $ is $O(1/L^{[k\}})$, and it takes $O(\log L^{[k]})$ -time to
compute. therefore, $L^{[k]}$ is bounded from above by $O(\exp(poly(n)))$.

So we somehow modify the composition of the machine $M_{A}$. Instead of
running the simulation of $M$ only once before the amplification, we run the
subroutine simulating $M$ as the number of $a_{2}$ increases by one.

\subsection{RTM $M_{A}$}

We mainly modify how the amplification stage is composed with other stages.
Here, without loss of generality, we suppose $j_{0}$, the step at which the
third step ends, is large enough:
\begin{equation}
j_{0}\geq l^{12}. \label{j-0}%
\end{equation}
This is realized by doing an extra task irrelevant to the input. Also we
suppose all the trace of the second and their steps in the tape cells are
initialized. (This is realized by erasing the trace of computation by moving
the machine backward.). $M$ goes into the amplification stage only if $M$
accepts the input. Otherwise, no successor is defined.

\begin{remark}
$j_{0}$ is a complicated functions of the distribution of the $M$-cells, as we
skip the $A$-cells between $M$-cells during these computation. But it can be
bounded from below by its value in case that no $A$ cell is present. It may be
as large as $O(\exp(l))$.
\end{remark}

In the amplification stage, we use the three states $Q_{u,4}=\{q_{4,\kappa
},\kappa=0,...,2\}$ and the three tape alphabets $\Gamma_{2,A}=\{a_{\kappa
},\kappa=1,2,3\}$, as well as $Q_{u,2}^{\prime}$ and $Q_{u,3}^{\prime}$ ,
which has the same elements of states as $Q_{u,2}$ and $Q_{u,3}$,
respectively. Roughly, the head rewrites $a_{1}$ to $a_{3}$, and $a_{3}$ to
$a_{2}$. When $a_{1}\rightarrow a_{3}$ took place, it goes back to the
$\square$-ed cell, then enters subroutine that simulates the second and the
third stages, but using the states in $Q_{u,2}^{\prime}$ and $Q_{u,3}^{\prime
}$. Upon accepting the input, the head again shifts to the $a_{3}$-ed cell,
and continues the amplification stage.

When the third stage finishes, the head is brought to the $\square$-ed cell,
and
\begin{align*}
(m_{0},q_{3,acc})  &  \rightarrow(m_{1},q_{4,1}),\\
(\ast,\square)  &  \rightarrow(\ast,\square).
\end{align*}
Then it moves as follows:%
\[%
\begin{tabular}
[c]{l|llll|l}
& \multicolumn{4}{|c|}{$m_{0}$} & $m_{1}$\\\cline{1-5}
& $a_{1}$ & $a_{2}$ & $a_{3}$ & $\square$ & \\\hline
$q_{4,0}$ &  & $a_{2},q_{4,0}$ & $a_{2},q_{4,1}$ &  & $\rightarrow$\\
$q_{4,1}$ & $q_{4,2}a_{3}$ &  &  &  & $\rightarrow$\\
$q_{4,2}$ &  & $a_{2},q_{4,2}$ &  & $\square,q_{2,init}^{\prime}$ &
$\leftarrow$\\
$q_{3,acc}^{\prime}$ &  &  &  & $\square,q_{4,0}$ & $\leftarrow$%
\end{tabular}
\ \ \ \
\]

So $q_{4,0}$: right-shifting, leaves $a_{2}$ unchanged. Rewrite $a_{3}$ to
$a_{2}$, while updated to $q_{4,1}$.

\ \ \ $q_{4,1}$: right-shifting, leaves $a_{2}$ unchanged. Rewrite $a_{1}$ to
$a_{3}$, while updated to $q_{4,2}$.

\ \ \ $\ q_{4,2}$: left-shifting, \ leaves $a_{2}$ unchanged. Upon reading
$\square$, goes into the simulation of $M$ using the states in $Q_{u}^{\prime
}$.

When the subroutine finishes with the state $(m_{0},q_{3,acc}^{\prime})$ at
the $\square$-ed site, it changes to $(m_{1},q_{4,0})$: note it is
$(m_{1},q_{4,1})$ at the start of the amplification stage. After the starting
step, the state will not be in $(m_{1},q_{4,1})$ if the head is at the
$\square$-ed cell.

\subsubsection{Evaluation of the average rate of $a_{2}$}

Here we evaluate the time average of the number $N_{2}$ of the $a_{2}$-ed
cells with respect to the distribution $p_{j;\boldsymbol{x}}$ as of (\ref{pj}).

Recall we use the 1D lattice with the size $L$ which is divided into smaller
blocks, and each block is used as a simulator of the RTM $M_{A}$. So below we
denote the length of the tape by $L^{[1]}$ and not by $L$, to avoid
confusions. The number of $M$ cells not marked by $\square$ is denoted by
$L_{m}$. Also, we suppose $j_{0}$ is large compared with $L^{[1]}$:%
\begin{equation}
j_{0}\geq(L^{[1]})^{3}. \label{j-0-2}%
\end{equation}
This is justified by (\ref{j-0}), since typically $L^{[1]}\leq l^{4}$ as will
be demonstrated in Section\thinspace\ref{sec:good-rate-2}.\thinspace

The number $N_{2}(j)$ of $a_{\kappa}$ at the time step $j$ is a complicated
function of the initial configuration, but it is bounded in a certain interval
if (\ref{j-0}) is true and the number $L_{m}$ of $M$-cells takes a `typical'
value indicated by the inequality%
\begin{equation}
\alpha L^{[1]}-(L^{[1]})^{2/3}<L_{m}<\alpha L^{[1]}+(L^{[1]})^{2/3}.
\label{Lm-range}%
\end{equation}

Suppose the first increment of $N_{2}(j)$ occurs at the step
$j_{\mathrm{first}}$. Then since the simulation of $M$ runs twice and the head
changes the direction of the shift only once in the amplification stage,
\begin{equation}
2j_{0}\leq j_{\mathrm{first}}\leq2j_{0}+4L^{[1]}. \label{jf}%
\end{equation}
Similarly, if $\Delta J_{2}(k)$ is the duration between the two increment of
$N_{2}$ when $N_{2}=k$,
\begin{equation}
j_{0}\leq\Delta J_{2}(k)\leq j_{0}+4L^{[1]}. \label{dJ}%
\end{equation}
Observe
\begin{align*}
J_{\boldsymbol{x}}  &  =j_{\mathrm{first}}+\sum_{k=1}^{L^{[1]}-L_{m}-2}\Delta
J_{2}(k),\\
\sum_{j=1}^{J_{\boldsymbol{x}}}N_{2}(j)  &  =\sum_{k=1}^{L^{[1]}-L_{m}%
-2}k\,\Delta J_{2}(k).
\end{align*}
By (\ref{jf}), (\ref{dJ}), and (\ref{Lm-range}),%

\[
(1-\alpha)L^{[1]}j_{0}-(L^{[1]})^{2/3}j_{0}\leq J_{x}\leq(1-\alpha
)L^{[1]}j_{0}+(L^{[1]})^{2/3}j_{0}+4(L^{[1]})^{2},
\]

\begin{align*}
\sum_{j=1}^{J_{\boldsymbol{x}}}N_{2}(j)  &  \geq\frac{j_{0}}{2}((1-\alpha
)L^{[1]}-(L^{[1]})^{2/3}-2)((1-\alpha)L^{[1]}-(L^{[1]})^{2/3}-1),\\
\sum_{j=1}^{J_{\boldsymbol{x}}}N_{2}(j)  &  \leq\frac{j_{0}}{2}((1-\alpha
)L^{[1]}+(L^{[1]})^{2/3}-2)((1-\alpha)L^{[1]}+(L^{[1]})^{2/3}-1)+4(L^{[1]}%
)^{3}.
\end{align*}
Therefore, by (\ref{j-0-2}),%

\begin{equation}
\lim_{L^{[1]}\rightarrow\infty}\frac{1}{J_{\boldsymbol{x}}(L^{[1]}+1)}%
\sum_{j=1}^{J_{\boldsymbol{x}}}N_{2}(j)=\frac{(1-\alpha)^{2}/2}{(1-\alpha
)}=\frac{1}{2}(1-\alpha), \label{mean-N2}%
\end{equation}
and \textit{the convergence is uniform for all the `good' configurations,
since both the upper and lower bound converges to the same value.}

As $N_{1}(j)+N_{2}(j)=L^{[1]}-L_{m}-1$,%
\begin{align}
&  \lim_{L^{[1]}\rightarrow\infty}\frac{1}{J_{\boldsymbol{x}}(L^{[1]}+1)}%
\sum_{j=1}^{J_{\boldsymbol{x}}}N_{1}(j)\nonumber\\
&  =\lim_{L^{[1]}\rightarrow\infty}\frac{L^{[1]}-L_{m}-1}{L^{[1]}+1}%
-\lim_{L^{[1]}\rightarrow\infty}\frac{1}{J_{\boldsymbol{x}}(L^{[1]}+1)}%
\sum_{j=1}^{J_{\boldsymbol{x}}}N_{2}(j)\nonumber\\
&  =1-\alpha-\frac{1}{2}(1-\alpha)=\frac{1}{2}(1-\alpha). \label{mean-N1}%
\end{align}
It is clear that the number $N_{3}$ of the cell with $a_{3}$ is at most $1$,
so
\begin{equation}
\lim_{L^{[1]}\rightarrow\infty}\frac{1}{J_{\boldsymbol{x}}(L^{[1]}+1)}%
\sum_{j=1}^{J_{\boldsymbol{x}}}N_{3}(j)=0. \label{mean-N3}%
\end{equation}
Clearly, the convergence in (\ref{mean-N1}) and (\ref{mean-N3}) are uniform
for all the `good' configurations.

\subsection{The Hamiltonian and time evolution}

A legal initial configuration for IID-case is consisted with possibly more
than single blocks, and each block is in a legal configuration as of the
previous analysis: So its left end site is $e_{0}=(m_{0},q_{1,init})$, and
other sites are either $e_{1}=\left(  \varsigma_{A},a_{1}\right)  $ or
$\left(  b,s_{0}\right)  $. A legal configuration is a configuration obtained
by applying $U$ to a legal initial configuration for finitely many times.
Here, we make sure that no interaction between two blocks are included in $U$.
We already had confirmed that no term nontrivially acts jointly on the
right-shifted finite control site and the $\square$-ed site ($(\mathcal{H}%
^{Q_{m}}\otimes\mathcal{H}^{Q_{u,+}})_{i}\otimes\left(  \mathcal{H}%
^{\Gamma_{1}}\otimes\mathbb{C}\left\vert \square\right\rangle \right)  _{i+1}%
$. Also, we make sure no term nontrivially acting on two consecutive finite
control sites $\mathcal{H}_{i}^{Q}\otimes\mathcal{H}_{i+1}^{Q}$. So the finite
control site at the left end of the site does not interact with its right
neighbor. Also, if there is a pair $\left\vert e_{0}\right\rangle
_{i}\left\vert e_{0}\right\rangle _{i+1}$ in the initial configuration,
$\left\vert e_{0}\right\rangle _{i}$ does not change in time.

Consequently, the Hamiltonian $H$ does not contain interaction between the
blocks. So if $\boldsymbol{x}$ $=\boldsymbol{x}^{[1]}\boldsymbol{x}%
^{[2]}\cdots\boldsymbol{x}^{[K]}$ is a legal configuration, where each
$\boldsymbol{x}^{[k]}$ is the $k$-th block,
\[
H\left\vert \boldsymbol{x}\right\rangle =\sum_{k=1}^{K}H^{[k]}\left\vert
\boldsymbol{x}\right\rangle ,
\]
where $H^{[k]}$ is the sum of terms acting only on $k$-th block. So
\[
e^{-itH}\,\left\vert \boldsymbol{x}^{[1]}\boldsymbol{x}^{[2]}\cdots
\boldsymbol{x}^{[K]}\right\rangle =\otimes_{k=1}^{K}e^{-\iota tH^{[k]}%
}\left\vert \boldsymbol{x}^{[k]}\right\rangle ,
\]
where the formula (\ref{hca-t}) applies to each $e^{-\iota tH^{[k]}}\left\vert
\boldsymbol{x}^{[k]}\right\rangle $. If $L^{[k]}+1$ is the size of the $k$-th
block,
\[
\left\langle \boldsymbol{x}\right\vert e^{itH}B^{(L)}e^{-itH}\,\left\vert
\boldsymbol{x}\right\rangle =\sum_{k=1}^{K}\frac{L^{[k]}+1}{L+1}\left\langle
\boldsymbol{x}^{[k]}\right\vert e^{itH^{[k]}}B^{(L^{[k]})}e^{-itH^{[k]}%
}\,\left\vert \boldsymbol{x}^{[k]}\right\rangle .
\]

In case of the periodic boundary condition, $\boldsymbol{x}^{[1]}$ is the
block containing the $0$-th site.

\subsection{The initial quantum state and good configurations}

\label{sec:good-rate-2}

Denote by $l$ the space necessary to simulate $M_{A}$ corresponding to the
input $v$, that is exponential in $n=\left\vert v\right\vert $. Define%
\begin{equation}
\left\vert \psi\right\rangle :=\sqrt{l^{-2}}\left\vert e_{0}\right\rangle
+\sqrt{1-l^{-2}}(\sqrt{1-\alpha}\left\vert \varsigma\right\rangle \left\vert
a_{1}\right\rangle +\sqrt{\alpha}\left\vert \,\mathrm{input}\,\right\rangle
\left\vert s_{0}\right\rangle ), \label{psi-2}%
\end{equation}
where $\left\vert \,\mathrm{input}\,\right\rangle $ is as of (\ref{input}).
\ Here recall $l$ is an upper bound to the space needed for the simulation of
the RTM $M_{A}$. As $l$ is exponentially large and $L$ is arbitrarily large
irrespective of $n$, we safely suppose that
\begin{equation}
\left(  L+1\right)  ^{1/11}\geq l\geq n^{6}, \label{L>l>n}%
\end{equation}
and (\ref{n>8}) as well.

We say an initial configuration $\boldsymbol{x}$ is `good' if the number of
sites covered by the `bad' blocks is at most%

\begin{equation}
3n^{-1}(L+1), \label{bad-blocks}%
\end{equation}
A block $\boldsymbol{x}^{[k]}$ is `good' iff it satisfies all of the following conditions:

\begin{description}
\item[(GB-1)] The length $L^{[k]}+1$ of \ $\boldsymbol{x}^{[k]}$ is not
shorter than $l$ and not longer than $l^{4}$

\item[(GB-2)] $\boldsymbol{x}^{[k]}$ $\ $satisfies (G-a,b)
\end{description}

\subsubsection{Evaluation of the `bad' rate}

For the sake of the notational simplicity, in this subsections, the definition
of `block' is slightly modified, in such a manner that does not affect the
result of the analysis. First, a block ends with $e_{0}$, rather than starting
from $e_{0}$. \ Second, regardless of the boundary condition, we suppose the
first block starts from the $0$-th site: so the true first block may be larger
than the first `block'.

For an initial configuration $\boldsymbol{x}$ to be a good', it suffices to
satisfy (G-c,d,e) below:

\begin{description}
\item[(G-c)] The lattice consists of approximately $K_{\ast}:=(1-l^{-2}%
)l^{-2}(L+1)$ blocks:
\[
(1-2l^{-2})(L+1)\leq\sum_{k=1}^{K_{\ast}}(L^{[k]}+1)\leq L+1.
\]

\item[(G-d)] Out of the sites covered by the first $K_{\ast}$ blocks, only
small part of them are covered by `bad' blocks:%
\[
\sum_{x^{[k]}:\text{-bad},1\leq k\leq K_{\ast}}(L^{[k]}+1)\leq(2l^{2}%
n^{-1}+3)K_{\ast}%
\]

\end{description}

At most $2l^{-2}(L+1)$ sites are not covered by the first the first $K_{\ast}%
$- blocks. For simplicity, suppose they are covered by `bad' blocks. Also,
recall the first `block' indeed may continue to the left of the first site, so
it may be too large and may be `bad'.

Therefore, the number of sites covered by the `bad' blocks is at most%

\begin{align*}
&  2l^{-2}(L+1)+(2l^{2}n^{-1}+3)K_{\ast}+l^{2}\\
&  =\left\{  2l^{-2}+2n^{-1}(1-l^{-2})+3(1-l^{-2})l^{-2}\right\}
(L+1)+l^{2}\\
&  \leq3n^{-1}(L+1).
\end{align*}
where we used (\ref{L>l>n}) and (\ref{n>8}). So (\ref{bad-blocks}) is satisfied.

Denote by $P_{\mathrm{good}}^{(c)}$ and $P_{\mathrm{good}}^{(d)}$, be the
projections onto the span of the configurations satisfying (G-c) and (G-d),
respectively, and denote by $P_{\mathrm{good}}$ the projection satisfying all
of them. Since they commute with each other,
\begin{align}
\,\mathrm{tr}\,P_{\mathrm{good}}\,\rho^{L}  &  \geq1-\sum_{\kappa
=c,d}(1-\mathrm{tr}\,P_{\mathrm{good}}^{(\kappa)}\,\rho^{L})\nonumber\\
&  \geq1-\frac{4l^{10}}{L}-\frac{l^{10}}{L}\nonumber\\
&  \geq1-\frac{5l^{10}}{L} \label{rho-good-iid-tr}%
\end{align}
as computed in the sequel.

Below, we demonstrate (\ref{rho-good-iid-tr}). Observe only the diagonal
elements of $\rho^{L}$ constitute to $\mathrm{tr}\,P_{\mathrm{good}}%
^{(\kappa)}\,\rho^{L}$, which we regard as a probability distribution. Below,
$\Pr$ is the probability with respect to this distribution.

In the initial configuration, the site is in either in $e_{0}$ or $e_{1}$, and
the probability of the former is $l^{-2}$. A single block is made of the
composition of the sequence of $e_{1}$ followed by a single $e_{0}$. Therefore,%

\begin{align*}
\Pr\{L^{[k]}+1  &  =m\}=l^{-2}(1-l^{-2})^{m-1},\\
\mathrm{E}\,(L^{[k]}+1)  &  =\sum_{m=0}^{\infty}m\Pr\{L^{[k]}=m\}=l^{2},\\
\mathrm{V}[L^{[k]}+1]  &  \leq\mathrm{E\,}[(L^{[k]}+1)^{2}]=2l^{4}%
(1-l^{-2})+l^{2}\\
&  \leq3l^{4}%
\end{align*}

Therefore, by Chebychev's inequality,%

\begin{align}
1-\mathrm{tr}\,P_{\mathrm{good}}^{(c)}\,\rho^{L}  &  =\Pr\left\{  \left\vert
\sum_{k=1}^{K_{\ast}}(L^{[k]}+1)-(1-l^{-2})(L+1)\right\vert \geq K_{\ast
}l^{-2}\right\} \nonumber\\
&  =\Pr\left\{  \left\vert \sum_{k=1}^{K_{\ast}}(L^{[k]}+1)-K_{\ast}%
\mathrm{E}(L^{[k]}+1)\right\vert \geq K_{\ast}l^{-2}\right\} \nonumber\\
&  \leq\frac{l^{4}}{K_{\ast}^{2}}\cdot3l^{4}K_{\ast}=\frac{3l^{8}}{K_{\ast}%
}=\frac{3l^{10}}{1-l^{-2}}\frac{1}{L+1}\nonumber\\
&  \leq\frac{4l^{10}}{L}, \label{p-c-good}%
\end{align}
where the last inequality is true if $l\geq2$.

A block with $l\leq L^{[k]}+1\leq l^{4}$ becomes `bad' with probability not
more than $p_{e}(n,l-1)$ by (\ref{p-good}). So if $W^{k}$ is the length of a
`bad' block,%

\[
W^{k}:=\left\{
\begin{array}
[c]{cc}%
0, & l\leq L^{[k]}+1\leq l^{4}\text{ and satisfies (G-a,b)}\\
L^{[k]}+1, & l\leq L^{[k]}+1\leq l^{4}\text{ and does not satisfy (G-a,b)}\\
L^{[k]}+1, & \text{otherwise,}%
\end{array}
\right.
\]

Here observe $l-1\geq n^{6}-1\geq2n^{3}$ by (\ref{L>l>n}), so by (\ref{L>n})
the probability for a block to satisfy (G-a,b) is bounded from above using
(\ref{p-good}). Therefore,
\begin{align*}
\mathrm{E}\,\left[  W^{k};l\leq L^{[k]}+1\leq l^{4}\right]   &  \leq
2n^{-1}\mathrm{E}\,[L^{[k]}+1;l\leq L^{[k]}+1\leq l^{4}]\\
&  \leq2n^{-1}\mathrm{E}\,[L^{[k]}+1]=2n^{-1}l^{2}.
\end{align*}
Also,
\begin{align*}
\mathrm{E\,}\left[  W^{k};L^{[k]}+1\leq l-1\text{ }\right]   &  \leq
\mathrm{E}\,\left[  L^{[k]}+1;L^{[k]}+1\leq l-1\text{ }\right] \\
&  \leq(l-1)\cdot\Pr\{L^{[k]}+1\leq l-1\text{ }\}\\
&  =(l-1)(1-(1-l^{-2})^{l})\\
&  \leq l(1-(1-l\cdot l^{-2}))=1,
\end{align*}
and%
\begin{align*}
\mathrm{E}\,\left[  W^{k};L^{[k]}+1\geq l^{4}+1\right]   &  =\mathrm{E}%
\,\left[  L^{[k]}+1;L^{[k]}+1\geq l^{4}+1\right] \\
&  =(1-l^{-2})^{l^{4}}(l^{2}+l^{4})\\
&  \leq e^{-l^{2}}(l^{2}+l^{4})\leq1,
\end{align*}
where the last inequality is by $l\geq1$. Summing all of them,
\[
\mathrm{E}\,[W^{k}]\leq2n^{-1}l^{2}+2.
\]

Also,
\[
\mathrm{V}[W^{k}]\leq\mathrm{E\,}[(W^{k})^{2}]\leq\mathrm{E\,}[(L^{[k]}%
+1)^{2}]\leq3l^{4},
\]
and by the Chebychev's inequality,
\begin{align*}
&  \Pr\left\{  \sum_{k=1}^{K_{\ast}}W^{k}\geq(2n^{-1}l^{2}+2)K_{\ast}+K_{\ast
}\right\} \\
&  \leq\Pr\left\{  \sum_{k=1}^{K_{\ast}}W^{k}\geq\mathrm{E}[W^{k}]K_{\ast
}+K_{\ast}\right\} \\
&  \leq\frac{3l^{4}}{K_{\ast}}=\frac{3l^{4}}{(1-l^{-2})l^{-2}(L+1)}\leq
\frac{l^{10}}{L}.
\end{align*}

\subsection{"Dephasing" and approximation}

Here we show that the expectation does not vary with cross terms between two
configurations that differs in a point of the separation of blocks.

\begin{lemma}
\label{lem:dephase-2}Suppose $B$ is an observable such that $B\left\vert
x\right\rangle =0$ unless $x=(\varsigma_{A},a_{\kappa})$. Suppose
$\boldsymbol{x}$ and $\boldsymbol{x}^{\prime}$ are legal initial
configurations. Then
\[
\left\langle \boldsymbol{x}^{\prime}\right\vert e^{tH}B^{(L)}e^{-\iota
tH}\left\vert \boldsymbol{x}\right\rangle \neq0
\]
only if the split into blocks occurs at the same points in $\boldsymbol{x}$
and in $\boldsymbol{x}^{\prime}$.
\end{lemma}

\begin{proof}
Observe%
\[
e^{-itH}\,\left\vert \boldsymbol{x}^{[1]}\boldsymbol{x}^{[2]}\cdots
\boldsymbol{x}^{[K]}\right\rangle \in\,\mathrm{span}\{\otimes_{k=1}%
^{K}\left\vert j^{(k)};\boldsymbol{x}^{[k]}\right\rangle \}.
\]
So we discuss the condition for
\[
\left\{  \otimes_{k=1}^{K^{\prime}}\left\langle j^{\prime(k)};\boldsymbol{x}%
^{\prime\lbrack k]}\right\vert \right\}  B^{(L)}\left\{  \otimes_{k=1}%
^{K}\left\vert j^{(k)};\boldsymbol{x}^{[k]}\right\rangle \right\}  \neq0.
\]
Denote the configuration corresponding to $\otimes_{k=1}^{K}\left\vert
j^{(k)};\boldsymbol{x}^{[k]}\right\rangle $ and $\otimes_{k=1}^{K}\left\vert
j^{\prime(k)};\boldsymbol{x}^{\prime\lbrack k]}\right\rangle $ by
$\boldsymbol{y}$ and $\boldsymbol{y}^{\prime}$, respectively.

First we show $y_{i}^{\prime}=y_{i}$ if $y_{i}\neq(\varsigma_{A},a_{\kappa})$
and $\left\langle \boldsymbol{y}^{\prime}\right\vert B^{(L)}\left\vert
\boldsymbol{y}\right\rangle \neq0$. Suppose $y_{i_{0}}\neq(\varsigma
_{A},a_{\kappa})$ and $y_{i_{0}}^{\prime}\neq y_{i_{0}}$. Since $B_{i}$
($i\neq i_{0}$) acts trivially on $\mathcal{H}_{i_{0}}$, $\left\langle
\boldsymbol{y}^{\prime}\right\vert B_{i}\left\vert \boldsymbol{y}\right\rangle
=0$. Moreover, $B\left\vert y_{i_{0}}\right\rangle =0$ by the hypothesis of
the lemma, so $\left\langle \boldsymbol{y}^{\prime}\right\vert B^{(L)}%
\left\vert \boldsymbol{y}\right\rangle =0$, contradicting the assumption.
Therefore, $y_{i}^{\prime}=y_{i}$ if $y_{i}\neq(\varsigma_{A},a_{\kappa})$.
Exchanging $\boldsymbol{y}$ and $\boldsymbol{y}^{\prime}$, we have that the
position and content of finite control sites are identical. Moreover, the
positions and the content of the sites marked by $\square$ are identical.

Second we show the split of the block occurs at the same points in
configurations. Suppose the $k$-th block in $\boldsymbol{y}$ starts from the
$y_{i_{0}}$.Then $y_{i_{0}}=(\ast,\square)$, or $q\in Q$.

Suppose $y_{i_{0}}=q\in Q$. Then $y_{i_{0}+1}=(\ast,\square)$, and $y_{i_{0}%
}^{\prime}=q\in Q$, $y_{i_{0}+1}=(\ast,\square)$. As $y_{i_{0}}=q\in Q$ is the
left end, $y_{i_{0}^{\prime}}=q\in Q$ cannot be the right end of a block: if
it were the case, $y_{i_{0}^{\prime}}=(m_{0},q^{\prime})$, $q^{\prime}\in
Q_{u,+}$, and $y_{i_{0}}=(m_{0},q^{\prime})$ could not be at the left end.
Therefore, $y_{i_{0}}^{\prime}$ is the left end of a block as well.

Next, suppose $y_{i_{0}}=(\ast,\square)$. Then $y_{i_{0}}^{\prime}=y_{i_{0}%
}=(\ast,\square)$. Unless $y_{i_{0}-1}^{\prime}=q\in Q$, $y_{i_{0}}^{\prime}$
is the left end of a block. Therefore, suppose $y_{i_{0}-1}^{\prime}=q\in Q$
as well. Then $y_{i_{0}-1}=y_{i_{0}-1}^{\prime}=q\in Q$. Since the former is
the former is the right end a block, so is the latter. Therefore, $y_{i_{0}%
}^{\prime}$ is the left end of a block.

As the positions of the split of the blocks are invariant, the split into
blocks occurs at the same points in $\boldsymbol{x}$ and in $\boldsymbol{x}%
^{\prime}$.
\end{proof}

Applying Lemma\thinspace\ref{lem:init-deco} to each block, we obtain:

\begin{lemma}
\label{lem:dephase-3} Let $\boldsymbol{x}$ and $\boldsymbol{x}^{\prime}$ be a
legal initial configurations. Suppose $B$ is an observable such that
$B\left\vert x\right\rangle =0$ unless $x=(\varsigma_{A},a_{\kappa})$. Then if
$\boldsymbol{x\neq x}^{\prime}$,%
\[
\,\left\langle \boldsymbol{x}^{\prime}\right\vert e^{\iota tH}B^{(L)}\text{
}e^{-\iota tH}\left\vert \boldsymbol{x}\right\rangle =0.
\]

\end{lemma}

If $\boldsymbol{x}$ is `good', there are $K$ ($\leq K_{\ast}$) of `good'
blocks that covers not less than $(L+1)-3n^{-1}(L+1)$ sites. So there are
$(1-\alpha-L^{-1/3})\{(L+1)-3n^{-1}(L+1)-K_{\ast}\}$ of $A$-cells at least.
But $K_{\ast}$ of them may be marked by $\square$. \ Therefore,%

\begin{align*}
\mathrm{tr}\,P_{E}\overline{\rho}_{\boldsymbol{x}}(t,L)  &  \geq\frac{1}%
{L+1}\{(1-\alpha-L^{-1/3})\{(L+1)-3n^{-1}(L+1)-K_{\ast}\}-K_{\ast}\}\\
&  \geq1-\alpha-L^{-1/3}-3n^{-1}-2l^{-2}.
\end{align*}
Therefore, in analogy to (\ref{error-rho-e-1-1}) and (\ref{error-rho-e-1-2}),
we obtain
\begin{align*}
\left\Vert \overline{\rho}(t,L)-\rho_{\ast}\right\Vert _{1}  &  \leq
2\sqrt{\alpha+L^{-1/3}+3n^{-1}+2l^{-2}}+\frac{5l^{10}}{L}+\max_{\boldsymbol{x}%
\text{:`good'}}\Vert P_{E}\,\overline{\rho}_{\boldsymbol{x}}(t,L)P_{E}%
-\rho_{\ast}\Vert_{1},\\
\left\Vert \lim_{T\rightarrow\infty}\frac{1}{T}\int_{t}^{T}dt\overline{\rho
}(t,L)-\rho_{\ast}\right\Vert _{1}  &  \leq2\sqrt{\alpha+L^{-1/3}%
+3n^{-1}+2l^{-2}}+\frac{5l^{10}}{L}\\
&  +\max_{\boldsymbol{x}\text{:`good'}}\Vert\lim_{T\rightarrow\infty}\frac
{1}{T}\int_{t}^{T}dt\,P_{E}\,\overline{\rho}_{\boldsymbol{x}}(t,L)P_{E}%
-\rho_{\ast}\Vert_{1}.
\end{align*}

Since
\[
\left\langle e_{\kappa^{\prime}}\right\vert \overline{\rho}_{\boldsymbol{x}%
}(t,L)\left\vert e_{\kappa}\right\rangle =\sum_{k=1}^{K}\frac{L^{[k]}+1}%
{L+1}\,\left\langle e_{\kappa^{\prime}}\right\vert \overline{\rho
}_{\boldsymbol{x}^{[k]}}(t,L^{[k]})\left\vert e_{\kappa}\right\rangle ,
\]
if $\boldsymbol{x=x}^{[1]}\boldsymbol{x}^{[2]}\cdots\boldsymbol{x}^{[K]}$ is a
`good' initial configuration,
\begin{align*}
\left\Vert P_{E}\overline{\rho}_{\boldsymbol{x}}(t,L)P_{E}-\rho_{\ast
}\right\Vert _{1}  &  =\sum_{k:1\leq k\leq K_{\ast},\boldsymbol{x}%
^{[k]}\text{:`good'}}\frac{L^{[k]}+1}{L+1}\,\left\Vert P_{E}\overline{\rho
}_{\boldsymbol{x}^{[k]}}(t,L^{\boldsymbol{x}^{[k]}})P_{E}-\rho_{\ast
}\right\Vert _{1}\\
&  +2\sum_{\boldsymbol{x}^{[k]}\text{:`bad'}}\frac{L^{[k]}+1}{L+1}\\
&  \leq\max_{\boldsymbol{x}^{[k]}\text{:`good'}}\,\left\Vert P_{E}%
\overline{\rho}_{\boldsymbol{x}^{[k]}}(t,L^{\boldsymbol{x}^{[k]}})P_{E}%
-\rho_{\ast}\right\Vert _{1}+3n^{-1},
\end{align*}
where we used (\ref{bad-blocks}). Therefore,
\begin{align}
\left\Vert \overline{\rho}(t,L)-\rho_{\ast}\right\Vert _{1}  &  \leq
2\sqrt{\alpha+L^{-1/3}+3n^{-1}+2l^{-2}}\nonumber\\
&  +\frac{5l^{10}}{L}+3n^{-1}+\max_{\boldsymbol{x}^{[k]}\text{:`good'}%
}\,\left\Vert P_{E}\overline{\rho}_{\boldsymbol{x}^{[k]}}(t,L^{\boldsymbol{x}%
^{[k]}})P_{E}-\rho_{\ast}\right\Vert _{1}. \label{error-rho-e-2-1}%
\end{align}

Analogously,
\begin{align}
\left\Vert \lim_{T\rightarrow\infty}\frac{1}{T}\int_{t}^{T}dt\overline{\rho
}(t,L)-\rho_{\ast}\right\Vert _{1}  &  \leq2\sqrt{\alpha+L^{-1/3}%
+3n^{-1}+2l^{-2}}+3n^{-1}+\frac{5l^{10}}{L}\nonumber\\
&  +\max_{\boldsymbol{x}^{[k]}\text{:`good'}}\,\Vert\lim_{T\rightarrow\infty
}\frac{1}{T}\int_{t}^{T}dt\,P_{E}\overline{\rho}_{\boldsymbol{x}^{[k]}%
}(t,L^{\boldsymbol{x}^{[k]}})P_{E}-\rho_{\ast}\Vert_{1}.
\label{error-rho-e-2-2}%
\end{align}
Here we can compute $\lim_{T\rightarrow\infty}\frac{1}{T}\int_{t}^{T}%
dt\,P_{E}\overline{\rho}_{\boldsymbol{x}^{[k]}}(t,L^{\boldsymbol{x}^{[k]}%
})P_{E}$ by (\ref{deco-step-2}).

\subsection{Proof of the second main lemma}

The proof of the second main lemma is almost parallel with the first one. Let
$M$ be an RTM which accepts an $\mathsf{EXPSPACE}$-complete problem such as
the problem of recognizing whether two regular expressions represents the same
regular language or not, where the expression is limited to union,
concatenation, * operation (zero or more copies of an expression), and
exponentiation (concatenation of an expression with itself $k$ times)
\cite{Sipser}. We use $n_{0}$ defined by (\ref{n>8}).

Let $S(n)$ be a polynomial time computable function of $n$ which is an
exponential of a polynomial of $n$ and is an upper bound to the space used by
the TM, and define
\begin{equation}
l(n):=\frac{1}{\alpha}\{\max\{S(n)+1,\,n^{2}\}+1\}^{3}. \label{l=S}%
\end{equation}
This satisfies (\ref{L>l>n}). Moreover, any `good' block has enough $M$-cells
to emulate the RTM $M$:%
\begin{align*}
(\alpha-l^{-1/3})l  &  \geq(\alpha-l^{--1/3})l^{1/3}\\
&  \geq\alpha l^{1/3}-1\\
&  \geq\max\{S(n)+1,\,n^{2}\}\geq S(n).
\end{align*}

If either the input length $n=|v|$ is too small or the input $v$ is rejected
by $M$, all the `good' blocks never flip the $A$-cell sites. If
$\boldsymbol{x}^{[k]}$ is a `good' block, by (\ref{rho-e1}) and the condition
(GB--1) on a 'good' block,%

\begin{align*}
\left\Vert P_{E}\overline{\rho}_{\boldsymbol{x}^{[k]}}(t,L^{\boldsymbol{x}%
^{[k]}})P_{E}-\left\vert e_{1}\right\rangle \left\langle e_{1}\right\vert
\right\Vert _{1}  &  \leq\alpha+(l-1)^{-1/3}+\frac{2}{l-1}\\
&  \leq\alpha+2\,l^{-1/3}.
\end{align*}

Therefore, by (\ref{error-rho-e-2-1}),%
\begin{align*}
\left\Vert \overline{\rho}(t,L)-\left\vert e_{1}\right\rangle \left\langle
e_{1}\right\vert \right\Vert _{1}  &  \leq2\sqrt{\alpha+L^{-1/3}%
+3n^{-1}+2l^{-2}}+\frac{5l^{10}}{L}+3n^{-1}+\alpha+2\,l^{-1/3}\\
&  \leq4\sqrt{\alpha},
\end{align*}
where the inequality in the second line is by (\ref{L>l>n}), (\ref{n>8}) and
(\ref{alpha}). \ Therefore, we only have to define define $L_{0}:=l^{11}$,
which is computable in polynomial time.

Suppose $M$ accept the input $v$. If $\boldsymbol{x}^{[k]}$ is a `good' block,
by the condition (GB--1) and (\ref{j-0}), (\ref{j-0-2}) is justified.
(\ref{Lm-range}) is justified by (GB-2). Therefore, by (\ref{mean-N2}),
(\ref{mean-N1}) and (\ref{mean-N3}), (recall the convergence is uniform in
them), \ for any $\varepsilon>0$, there is an $l_{0}$ such that for any $l\geq
l_{0}$ (recall $L^{[k]}\geq l-1$ if $\boldsymbol{x}^{[k]}$ is `good'),%

\begin{align*}
\left\vert \frac{1}{J_{\boldsymbol{x}^{[k]}}(L^{[x]}+1)}\sum_{j=1}%
^{J_{\boldsymbol{x}^{[k]}}}N_{\kappa}(j;\boldsymbol{x}^{[k]})-\frac{1}%
{2}\right\vert  &  \leq\frac{1}{2}\alpha+\varepsilon,\,(\kappa=1,2),\\
\frac{1}{J_{\boldsymbol{x}^{[k]}}(L^{[x]}+1)}\sum_{j=1}^{J_{\boldsymbol{x}%
^{[k]}}}N_{3}(j;\boldsymbol{x}^{[k]})  &  \leq\varepsilon,\,(\kappa=1,2).
\end{align*}
So by (\ref{deco-step-2}),
\[
\,\max_{\boldsymbol{x}^{[k]}\text{:`good'}}|\lim_{T\rightarrow\infty}\frac
{1}{T}\int_{t}^{T}dt\,\left\langle e_{\kappa}\right\vert \overline{\rho
}_{\boldsymbol{x}^{[k]}}(t,L^{\boldsymbol{x}^{[k]}})\left\vert e_{\kappa
^{\prime}}\right\rangle -\frac{1}{2}|\leq\left\{
\begin{array}
[c]{cc}%
\frac{1}{2}\alpha+\varepsilon, & \kappa=\kappa^{\prime}=1\,\text{or }2,\\
\varepsilon & \text{otherwise}%
\end{array}
\right.  ,
\]
and
\begin{align*}
\,\max_{\boldsymbol{x}^{[k]}\text{:`good'}}\,\Vert\lim_{T\rightarrow\infty
}\frac{1}{T}\int_{t}^{T}dt\,P_{E}\overline{\rho}_{\boldsymbol{x}^{[k]}%
}(t,L^{\boldsymbol{x}^{[k]}})P_{E}-\frac{1}{2}(\left\vert e_{1}\right\rangle
\left\langle e_{1}\right\vert +\left\vert e_{2}\right\rangle \left\langle
e_{2}\right\vert )\Vert_{1}  &  \leq2\cdot(\frac{1}{2}\alpha+\varepsilon
)+7\varepsilon\\
&  =\alpha+9\varepsilon.
\end{align*}
By (\ref{error-rho-e-2-2}).%
\begin{align}
\lim_{L\rightarrow\infty}\left\Vert \lim_{T\rightarrow\infty}\frac{1}{T}%
\int_{t}^{T}dt\overline{\rho}(t,L)-\frac{1}{2}(\left\vert e_{1}\right\rangle
\left\langle e_{1}\right\vert +\left\vert e_{2}\right\rangle \left\langle
e_{2}\right\vert )\right\Vert _{1}  &  \leq2\sqrt{\alpha+3n^{-1}+2l^{-2}%
}+3n^{-1}+\alpha+9\varepsilon\nonumber\\
&  \leq4\sqrt{\alpha}+9\varepsilon,
\end{align}
where the second inequality is by (\ref{n>8}), (\ref{L>l>n}), and
(\ref{alpha}). As $\varepsilon$ is arbitrary, the proof completes.

\section{More on the initial state (\ref{rho-L}) and (\ref{rho-L-iid})}

So far we had supposed the function $f_{\psi}$ can be computed with the error
at most $\varepsilon$ with time in polynomial of $\left\vert v\right\vert $
and $\log(1/\varepsilon)$. Here we relax the condition, and suppose it is
simply computable up to the arbitrarily specified accuracy by a total
recursive function of $|v|$ and $\left\lceil \log(1/\varepsilon)\right\rceil
$, or a function computed by a TM that halts on any input. Clearly, this
variant of \textsf{SAS} and \textsf{SAH} are \textsf{RE}-complete. So let us
consider the variant of \textsf{SAS-iid} and \textsf{SAH-iid}. Because these
two are essentially the same, below we only discuss the former, and denote it
by \textsf{SAS-iid'}.

Denote by $\mathsf{R}$ the set of decidable problems, or equivalently,
decision problems solved by a TM that halts on any input. We argue
\[
\mathsf{R}\subset\{\text{\textsf{SAS-iiid'}}(d,H,f_{\psi},\frac{1}{2},\frac
{1}{4});d\in\mathbb{N},H,f_{\psi}\},
\]
where $H$ and $f_{\psi}$ runs over all the shift-invariant Hamiltonians with
nearest neighbor interactions and computable functions, respectively.

To this end, we show any $P\in$ $\mathsf{R}$ equals \textsf{SAS-iiid'}%
$(d,H,f_{\psi},\frac{1}{2},\frac{1}{4})$ for some $d$, $H$, and $f_{\psi}$.

Suppose a decision problem $P$ is an instance of $\mathsf{R}$. Then there is
an RTM $M$ that solves $P$ using the space $S(n)$, where $n:=|v|$. Clearly,
the function $S(\cdot)$ can be computed by a TM $M^{\prime}$ that halts on any
$n$: \ Given $n$, it generates a bit string $v$ with $|v|\leq n$, simulate $M$
on it, count the space used by the computation, and take maximum of them for
all $v$ with $|v|\leq n$. So $f_{\psi}$ defined by (\ref{psi-2}) and
(\ref{l=S}) is a computable function. So by the proof of the second main
lemma, the output of the TM $M$ equals the output of the corresponding
\textsf{SAS-iiid'}$(d,H,f_{\psi},\frac{1}{2},\frac{1}{4})$.

It is not clear whether \textsf{SAS-iid} and \textsf{SAH-iid} are decidable
for some set of the parameters. If we use the composition used in the proof of
the second main lemma, it can emulate a TM that solves an instance of
$\mathsf{R}$. However, there can be some dynamics whose behavior in
$L\rightarrow\infty$ cannot be computed by an TM, even if the initial state is
as simple as (\ref{rho-L-iid}).

\section{On finite-size lattices and computational complexity}

\subsection{Problem and settings}

Hereafter, the lattice size $L$ is finite and given as an input. The length of
the input $L$ should be defined to discuss computational complexity: so be the
number of bits in binary expansion, or $L$ itself. First, we consider the
former setting. In this setting, the finite lattice version of $\mathsf{SAH}$
is clearly $\mathsf{EXPSPACE}$-hard by the second main lemma.

We show they are in fact $\mathsf{EXPSPACE}$-complete. Our strategy is to find
an algorithm to compute either $\lim_{T\rightarrow\infty}\frac{1}{T}\int
_{t}^{T}dt\overline{\rho}(t,L)$ or $\overline{\rho}(t,L)$ using exponential
space. The accuracy achieved by finite resource is finite: therefore, a the
parameter $\varepsilon_{1}$, that is linked to the accuracy, should be a
function of input length.

So we formulate the problem as follows:

\begin{description}
\item[{$\mathbf{[[}\mathsf{SAHF}(d,f_{H},\gamma,\eta,\varepsilon
_{1})]\mathbf{]}$}] 

\item[Fixed:] $d=\dim\mathcal{H}$, $\{\left\vert e_{\kappa}\right\rangle
;\left\langle e_{\kappa^{\prime}}\right.  \left\vert e_{\kappa}\right\rangle
=\delta_{\kappa,\kappa^{\prime}}\}_{\kappa=0,1,2}$, $\left\vert \psi
\right\rangle :=\left\vert e_{1}\right\rangle $, real numbers $\varepsilon
_{2}$ and $\eta$ with $0<2\varepsilon_{1}<\eta<1$. The function $f_{H}$ of $v$
to the 1- and 2- body term of the Hamiltonian $H$ such that: $\overline{\rho
}(t,L)$ satisfies either (\ref{sa-1}) or (\ref{sa-2}), and%
\begin{equation}
\left\vert \lambda-\lambda^{\prime}\right\vert \geq2^{-L^{\gamma}}.
\label{sp-gap}%
\end{equation}
Here $\lambda$ and $\lambda^{\prime}$ are distinct eigenvalues of $PHP$, where
$P$ is the projection onto the smallest invariant subspace of $H$ containing
the initial state vector $\left\vert \psi\right\rangle ^{\otimes L+1}$.
\ Also, $\gamma\in\mathbb{N}$. Moreover, it is computable with the error at
most $\varepsilon$ using time in polynomial of the input length $n:=|\nu
|+\lceil\log L\rceil$.

\item[Input:] A natural number $L$ and a bit string $v$.

\item[Question:] (\ref{sa-1}) is true or not.
\end{description}

\begin{theorem}
\label{thm:expspace}Suppose $d:=\dim$ $\mathcal{H}$ is fixed and larger than a
certain threshold $d_{0}$. \ For any $\{\left\vert e_{\kappa}\right\rangle
\}_{\kappa=0,1,2}$ $\mathsf{SAHF}(d,f_{H},\frac{1}{2},\varepsilon_{1})$ with
$\varepsilon_{1}\in(0,1/4)$ is $\mathsf{EXPSPACE}$-complete.
\end{theorem}

The proof that the problem is $\mathsf{EXPSPACE}$-hard is almost analogous to
it of Theorem\thinspace\ref{thm:SA}, so we only sketch this part of the proof:
The Hamiltonian constructed in showing the main lemmas has the energy gap
$O(1/J_{\boldsymbol{x}})^{2}$, and $J_{\boldsymbol{x}}=O(\exp(p(l)))\leq
O(\exp(p(L)))$ at most ( $p$ is a polynomial), since the TM using space $p(l)$
runs for at most $\exp(c\cdot p(l))$ steps ($c$ is a constant). Therefore, the
spectrum of the Hamiltonian $H=f_{H}(v)$ satisfies the condition
(\ref{sp-gap}). So we obtain an analogue of the second main lemma, and
$\mathsf{EXPSPACE}$-hardness follows from it.

From the next subsection, we show the problem is in $\mathsf{EXPSPACE}$ by
computing $\overline{\rho}(t,L)$ using exponential space.

\subsection{Computational complexity of linear algebraic operations}

\ The elementary arithmetic and linear algebraic operations can be done by
shallow circuits, and the work space necessary for a TM to simulate a circuit
is related to the depth and the size of the circuit.

Denote by $\mathsf{NC}(s)$ the class of problems that can be solved by space
$O(s)$-\thinspace uniform boolean circuits having size $2^{O(s)}$ and depth
$s^{O(1)}$, where $s(n)$ is any function with $s(n)\geq\log n$. It is known
that elementary arithmetic of $s$-digit numbers, matrix multiplication,
addition, and computation of $\left\Vert A\right\Vert _{1}=\,\mathrm{tr}%
\,\sqrt{A^{\dagger}A}$ of $s\times s$ matrices are all contained in
$\mathsf{NC}(p(s))$, where $p(s)$ is a polynomial function of $s$. So the
composition of these operations for $s^{O(1)}$ times is also contained in
$\mathsf{NC}(2^{s})$. Moreover, it is known that the class $\mathsf{NC}%
(2^{s})$ is contained in $\mathsf{DSPACE}(s^{O(1)})$, which is the class of
problems solved by a deterministic Turing machine with space $s^{O(1)}$.

Therefore, \ linear algebraic operations of $d^{L+1}\times d^{L+1}$- matrices
with the error at most $O(2^{-L^{\gamma}})$ (accuracy up to $O(L^{\gamma})$
digits) are contained in $\mathsf{NC}(p(2^{L}))$ for some polynomial function
$p$, so in $\mathsf{DSPACE}(p^{\prime}(L))$ for some polynomial function
$p^{\prime}$. Since $L$ is an exponential function of the input length $n$,
$\mathsf{DSPACE}(p^{\prime}(L))$\ is contained in $\mathsf{EXPSPACE}$.

\subsection{Proof of Theorem\thinspace\ref{thm:expspace}}

The following argument is more or less similar to it of Sec.\thinspace
\ref{sec:re-complete}: we check the condition (\ref{check-cond-approx})
against the alternative (\ref{sa-2}) at all $K$ with $K\leq\left\lceil
T_{0}/(t_{i}-t_{i-1})\right\rceil $, where $t_{i}$ is as of (\ref{dt}), \
\begin{align*}
T_{0}  &  :=2^{2(L+1)+2L^{\gamma}+1},\\
1  &  \leq i\leq\left\lceil \frac{4T_{0}\left\Vert H\right\Vert }%
{\eta-\varepsilon_{1}}\right\rceil =O(L\,2^{2L}2^{2L^{\gamma}}).
\end{align*}
This discretization of the time is justified by the argument in Sec.\thinspace
\ref{sec:re-complete}. The cut-off $T\leq T_{0}$ is justified by%

\begin{align*}
\left\Vert \lim_{T\rightarrow\infty}\frac{1}{T}\int_{0}^{T}dt\,\overline{\rho
}(t,L)-\frac{1}{T_{0}}\int_{0}^{T_{0}}dt\,\overline{\rho}(t,L)\right\Vert
_{1}  &  \leq2^{2(L+1)}\frac{1}{T_{0}}\left\vert \int_{0}^{T_{0}}%
e^{\iota2^{-L^{\gamma}}t}dt\right\vert \\
&  =\frac{2^{2(L+1)}2^{L^{\gamma}}}{T_{0}}\left\vert e^{-\iota2^{-\gamma}%
T_{0}}-1\right\vert \leq2^{-L^{\gamma}},
\end{align*}
where we had used that the energy gap is at least $2^{-L^{\gamma}}$.

Given the out come of each verification, the problem is computed by a
additional circuit having size $p(T_{0}\left\Vert H\right\Vert )$ and depth
$p^{\prime}(\log T_{0}\left\Vert H\right\Vert )$ circuit ($p$ and $p^{\prime}$
are polynomials.).

To check the condition (\ref{check-cond-approx}) against the alternative
(\ref{sa-2}), it suffices to compute $\left\Vert \overline{\rho}%
(t_{i},L)-\left\vert e_{1}\right\rangle \left\langle e_{1}\right\vert
\right\Vert _{1}$ with the error at most $\frac{3}{16}(\eta-\varepsilon_{1})$.
To this end, we approximate $\rho^{L}(t_{i})$ by
\[
\rho_{\mathrm{ap}\,}^{L}:=\sum_{k=1}^{2N^{2}}\frac{(-\iota tH_{\mathrm{ap}%
\,})^{k}}{k!}\rho^{L}\sum_{k=1}^{2N^{2}}\frac{(\iota tH_{\mathrm{ap}\,})^{k}%
}{k!},
\]
where $H_{\mathrm{ap}\,}$ is a numerical approximation of $H$ with the error
at most
\begin{align*}
\left\Vert H-H_{\mathrm{ap}\,}\right\Vert  &  \leq\frac{\eta-\varepsilon_{1}%
}{16T_{0}},\\
H_{\mathrm{ap}\,} &  :=\sum(H_{\mathrm{ap}\,,i}+H_{\mathrm{ap}\,,i,i+1})\\
&  \left\Vert H_{\mathrm{ap}\,,i}-H_{i}\right\Vert +\left\Vert H_{\mathrm{ap}%
\,,i,i+1}-H_{i,i+1}\right\Vert \leq\frac{\eta-\varepsilon_{1}}{16T_{0}(L+1)},
\end{align*}
and%
\[
N:=\left\lceil T_{0}\left\Vert H_{\mathrm{ap}\,}\right\Vert \right\rceil
\leq\left\lceil T_{0}\left\Vert H\right\Vert \right\rceil +\frac
{\eta-\varepsilon_{1}}{16T_{0}}\leq\left\lceil T_{0}\left\Vert H\right\Vert
\right\rceil +1.
\]
Then as we demonstrate soon,%

\begin{align}
\left\Vert \overline{\rho}_{\mathrm{ap}}(t_{i},L)-\overline{\rho}%
(t_{i},L)\right\Vert _{1}  &  \leq\left\Vert \rho_{\mathrm{ap}\,}^{L}%
(t)-\rho_{\,}^{L}(t)\right\Vert _{1}\nonumber\\
&  \leq\frac{5}{2}\{2^{-(N^{2}-N)}+\frac{\eta-\varepsilon_{1}}{16}\}.
\label{num-error}%
\end{align}
So we can check the condition (\ref{check-cond-approx}) by computing
$\left\Vert \overline{\rho}_{\mathrm{ap}}(t_{i},L)-\left\vert e_{1}%
\right\rangle \left\langle e_{1}\right\vert \right\Vert _{1}$ with the error
at most $\frac{\eta-\varepsilon_{1}}{16\cdot4}$, which is easy (see Sec.6,
\cite{Wat2003}).

Below, we show (\ref{num-error}). Observe
\begin{align*}
\left\Vert \exp(-\iota tH_{\mathrm{ap}\,})-\exp(-\iota tH)\right\Vert  &
=\left\Vert \int_{0}^{1}ds\,e^{-\iota(1-s)tH}(-\iota t(H_{\mathrm{ap}%
\,}-H))e^{-\iota stH_{\mathrm{ap}\,}}\right\Vert \\
&  \leq\int_{0}^{1}ds\left\Vert e^{-\iota(1-s)tH}\right\Vert \left\Vert -\iota
t(H_{\mathrm{ap}\,}-H)\right\Vert \left\Vert e^{-\iota stH_{\mathrm{ap}\,}%
}\right\Vert \\
&  \leq\int_{0}^{1}ds\left\Vert T_{0}(H_{\mathrm{ap}\,}-H)\right\Vert
=T_{0}\left\Vert H_{\mathrm{ap}\,}-H\right\Vert \\
&  \leq T_{0}\cdot\frac{\eta-\varepsilon_{1}}{16T_{0}}=\frac{\eta
-\varepsilon_{1}}{16},
\end{align*}
where the identity in the first line is by X.4.2, p.311, \cite{Bahatia}. Also,%

\begin{align*}
&  \left\Vert \exp(-\iota tH_{\mathrm{ap}\,})-\sum_{k=1}^{2N^{2}}\frac{(-\iota
tH_{\mathrm{ap}\,})^{k}}{k!}\right\Vert \\
&  \leq\sum_{k=2N^{2}+1}^{\infty}\frac{\left(  T_{0}\left\Vert H_{\mathrm{ap}%
\,}\right\Vert \right)  ^{k}}{k!}=\sum_{k=2N^{2}+1}^{\infty}\frac
{T_{0}\left\Vert H_{\mathrm{ap}\,}\right\Vert }{1}\cdot\frac{T_{0}\left\Vert
H_{\mathrm{ap}\,}\right\Vert }{2}\cdot\ldots\cdot\frac{T_{0}\left\Vert
H_{\,\mathrm{ap}\,}\right\Vert }{k}\\
&  \leq\sum_{k=2N^{2}+1}^{\infty}\left(  \frac{T_{0}\left\Vert H_{\mathrm{ap}%
\,}\right\Vert }{1}\right)  ^{N}\cdot\left(  \frac{T_{0}\left\Vert
H_{\mathrm{ap}\,}\right\Vert }{T_{0}\left\Vert H_{\mathrm{ap}\,}\right\Vert
}\right)  ^{N^{2}-N-1}\cdot\left(  \frac{T_{0}\left\Vert H_{\mathrm{ap}%
\,}\right\Vert }{N^{2}}\right)  ^{k-N^{2}+1}\\
&  \leq\sum_{k=2N^{2}+1}^{\infty}\left(  T_{0}\left\Vert H_{\mathrm{ap}%
\,}\right\Vert \right)  ^{-(k-N^{2}-N+1)}=\frac{(T_{0}\left\Vert
H_{\mathrm{ap}\,}\right\Vert )^{-(N^{2}-N+2)}}{1-\left(  T_{0}\left\Vert
H_{\mathrm{ap}\,}\right\Vert \right)  ^{-1}}\\
&  \leq(T_{0}\left\Vert H_{\mathrm{ap}\,}\right\Vert )^{-(N^{2}-N)}%
\leq2^{-(N^{2}-N)},
\end{align*}
so%
\[
\left\Vert \exp(-\iota tH)-\sum_{k=1}^{2N^{2}}\frac{(-\iota tH_{\mathrm{ap}%
\,})^{k}}{k!}\right\Vert _{1}\leq2^{-(N^{2}-N)}+\frac{\eta-\varepsilon_{1}%
}{16}.
\]
Therefore, by triangle inequality we obtain (\ref{num-error}):%

\begin{align*}
&  \left\Vert \rho_{\,\mathrm{ap}\,}^{L}(t)-\rho^{L}(t)\right\Vert _{1}\\
&  \leq\left\Vert \sum_{k=1}^{2N^{2}}\frac{(-\iota tH_{\mathrm{ap}\,})^{k}%
}{k!}-e^{-\iota tH}\right\Vert \left\Vert \rho^{L}\right\Vert _{1}\left\Vert
e^{\iota tH}\right\Vert +\left\Vert \sum_{k=1}^{2N^{2}}\frac{(-\iota
tH_{\mathrm{ap}\,})^{k}}{k!}\right\Vert \left\Vert \rho^{L}\right\Vert
_{1}\left\Vert \sum_{k=1}^{2N^{2}}\frac{(-\iota tH_{\mathrm{ap}\,})^{k}}%
{k!}-e^{-\iota tH}\right\Vert \\
&  \leq2^{-(N^{2}-N)}+\frac{\eta-\varepsilon_{1}}{16}+(1+2^{-(N^{2}%
-N)})(2^{-(N^{2}-N)}+\frac{\eta-\varepsilon_{1}}{16})\\
&  \leq\frac{5}{2}\{2^{-(N^{2}-N)}+\frac{\eta-\varepsilon_{1}}{16}\}.
\end{align*}

Below we discuss computational cost for computing $\left\Vert \overline{\rho
}_{\mathrm{ap}}(t_{i},L)-\left\vert e_{1}\right\rangle \left\langle
e_{1}\right\vert \right\Vert _{1}$. In the following, $p$, $p^{\prime}%
$,$p^{\prime\prime}$ etc., are some polynomial functions.

The computation of constant size matrices $H_{i,i+1}$ and $H_{i}$ are done in
$O(p(\log T_{0}(L+1)))$, and addition of these terms ($i=1,\cdots,L$) is done
by a circuit having size $O(p(L))$ and depth $O(p(\log L))$. The $k$-th power
of the Hamiltonian is computed by a circuit having \ size $O(p(N))$ and depth
$O(p(\log N))$ using the matrix multiplication as a subroutine, which is
implemented by a circuit having size $O(p^{\prime}(d^{L},\log T_{0}(L+1)))$
and depth $O(p^{\prime}(L,\log\log T_{0}(L+1)))$ ($p$ and $p^{\prime}$ are
polynomials.). So the composition of them is computed by a circuit having size
$O(\exp(p^{\prime\prime}(L)))$ and depth $O(p^{\prime\prime}(L,\gamma))$

To sum the $k$-th powers from $k=1$ to $2N^{2}$, we use a circuit having
\ size $O(p(N))$ and depth $O(p(\log N)))$ using the matrix addition and
$k$-th power as a subroutine, which is implemented by a circuit having size
$O(\exp(p(L)))$ and depth $O(p(L))$ : so the composition can be computed by a
circuit having\ size $O(\exp(p^{\prime}(L)))$ and depth $O(p(L))$.

To compute $\overline{\rho}_{\mathrm{ap}}(t_{i},L)$, we multiply the result of
the above computation with $(\left\vert e_{1}\right\rangle \left\langle
e_{1}\right\vert )^{\otimes L+1}$, compute reduced state to the $i$-th site,
and add them from $1$ to $L+1$, and divide the resulting object by $L+1$: this
can be done using a circuit having size $O(\exp(p(L)))$ and depth $O(p(L))$,
that uses the subroutine to compute $\overline{\rho}_{\,\mathrm{ap}\,}%
(t_{i},L)$.

Since the computation of the trace distance can be done $O(\exp(p(L)))$, see
Sec.6, \cite{Wat2003}: \ Our case is in fact much easier, since the size of
the matrix is finite. The characteristic polynomial of the matrix can be
crudely computed, and the eigenvalues are computed by Neff's
algorithm\thinspace\cite{Neff}.

After all, the whole process can be done by a circuit having size
$O(\exp(p(L)))$ and depth $O(p(L))$, so it can be done by a Turing machine
having space $O(p(L))=O(\exp(p^{\prime}(n))$.

\subsection{\textsf{PSPACE}-completeness}

\subsubsection{Statement of the result and difficulties}

We modify the definition of $\mathsf{SAHF}(d,f_{H},\gamma,\eta,\varepsilon
_{1})$: the definition of input length is changed to $n:=\left\vert
v\right\vert +L$. This modified version is called $\mathsf{SAHF}^{\prime
}(d,f_{H},\gamma,\eta,\varepsilon_{1})$

\begin{theorem}
$\mathsf{SAHF}^{\prime}(d,f_{H},\gamma,\frac{1}{2},\varepsilon_{1})$
($\varepsilon_{1}\leq\frac{1}{4}$) is \textsf{PSPACE}-complete.
\end{theorem}

The proof that $\mathsf{SAHF}^{\prime}$ belongs to \textsf{PSPACE} is simply
scaling down the parameters of the proof of Theorem\thinspace
\ref{thm:expspace}. The proof of the hardness is almost analogous, except
except the second stage, where exponentially many bits are used. We use the
encoding of the input bit string that can be decoded using only polynomially
many qubits. It is based on a $1$-qubit version of phase estimation algorithm
which saves the space at the expense of the time complexity.

Different from the encoding used so far, the input is not encoded to the
eigenvector of the Hamiltonian, so it is not possible to modify the scheme to
the encoding to the initial state.

An apparent difficulty of this scheme involves operations (rotations of a
qubit) which creates superposition. We demonstrate, however, this difficulty
can be circumvented by flipping a bit classical bit at every rotation of a qubit.

\subsubsection{New encoding and decoding process}

We consider the algorithm that encode inputs both to he Hamiltonian and the
initial state. But by the argument used in the proof of Theorem \ref{thm:SA},
(ii), the informations encoded to the state is moved to the Hamiltonian.

The Hilbert space is slightly modified.
\begin{align*}
\mathcal{H}^{Q}  &  :=\mathcal{H}^{Q_{m}}\otimes\mathcal{H}^{Q_{u}}%
\otimes\mathcal{H}^{Q_{in}},\\
\mathcal{H}^{Q_{in}}  &  :=\mathrm{span}\,\{\left\vert \zeta_{0}\right\rangle
,\left\vert \zeta_{1}\right\rangle \},\\
\mathcal{H}^{\Gamma_{1,M}}  &  :=\mathrm{span}\,\{\left\vert b\right\rangle
;b=0,1\},
\end{align*}
and we add more symbols to $\Gamma_{2,M}$ to implement more complicated
decoding process. In addition, for notational simplicity, we add
$\mathcal{H}^{Q_{in}}$ to both $M$- and $A$-cells:
\[
\mathcal{H}^{\Gamma}=\{(\mathcal{H}^{\Gamma_{1,M}}\otimes\mathcal{H}%
^{\Gamma_{2,M}})\oplus(\mathcal{H}^{\Gamma_{1,A}}\otimes\mathcal{H}%
^{\Gamma_{2,A}})\}\otimes\mathcal{H}^{Q_{in}},
\]
though it is set to $\left\vert \zeta_{0}\right\rangle $ and never used.

The initial state is
\begin{align*}
\left\vert \psi\right\rangle  &  :=l^{-1}\left\vert e_{0}\right\rangle
+(1-l^{-2})^{1/2}\sqrt{1-\alpha}\left\vert e_{1}\right\rangle \\
&  +(1-l^{-2})^{1/2}\sqrt{\alpha}(|v|^{-1}\left\vert 1\right\rangle
+(1-|v|^{-2})^{1/2}\left\vert 0\right\rangle \left\vert s_{0}\right\rangle
),\\
\left\vert e_{0}\right\rangle  &  :=\left\vert m_{0}\right\rangle \left\vert
q_{1,init}\right\rangle \left\vert \zeta_{0}\right\rangle .
\end{align*}
so the separation of the blocks, the rate of $M$-cells, and an upper bound
$n^{\prime}$ to $|v|$ are encoded to the state. The input bit string is
encoded to the 1-body term of the Hamiltonian

Also, the Hamiltonian has the 1-body term $H_{i}^{in}=U_{i}^{in}+(U_{i}%
^{in})^{\dagger}$,
\begin{align*}
U^{in}  &  :=U^{in,1}+U^{in,2},\\
U^{in,1}  &  :=\left\vert m_{0}\right\rangle \left\langle m_{0}\right\vert
\otimes\sum_{\kappa=0,1}\left\vert q_{a,1},w_{\kappa}\right\rangle
\left\langle q_{b,1},w_{\kappa}\right\vert \otimes R_{\pi\beta},\\
U^{in,2}  &  :=\left\vert m_{0}\right\rangle \left\langle m_{0}\right\vert
\otimes\left\vert q_{a,2},w_{1}\right\rangle \left\langle q_{b,1}%
,w_{1}\right\vert \otimes R_{-\pi2^{-|v|}},
\end{align*}
where $R_{\theta}$ is the rotation by the angle $\theta$, the fractional part
of $\beta=0.v_{1}v_{2}\cdots v_{|v|}$ equals the input bit strings $v$ (with
the promise $v_{|v|}=1$). Also, $q_{b,1}$ ($q_{b,2}$, resp.) is the state that
indicates $R_{\pi\beta}$ ($R_{-\pi2^{-|v|}}$, resp.) to be applied to
$\mathcal{H}^{Q_{in}}$, and $q_{a,1}$ ($q_{a,2}$, resp.) indicates that the
rotation is just done. The role of $w=w_{0}$,$w_{1}$ be explained soon.

The Hamiltonian is based on the following quantum algorithm. $n^{\prime}$ is
encoded and decoded in the same manner as the previous procedure. Roughly, we
first decode $|v|$ and then $\beta$. Observe $\beta_{|v|}=1$ and $\beta
_{k_{1}}=0$ ($k_{1}>|v|$). Therefore, it holds that
\begin{equation}
(R_{\pi\beta})^{2^{|v|}}\left\vert \zeta_{0}\right\rangle =\pm\left\vert
\zeta_{1}\right\rangle ,\,\,\,\,\,\,(R_{\pi\beta})^{2^{k_{1}}}\left\vert
\zeta_{0}\right\rangle =\pm\left\vert \zeta_{0}\right\rangle \,\,\,(k_{1}%
>|v|), \label{|v|}%
\end{equation}
and by this we can decode $|v|$.

Also, let
\[
\beta^{(k_{1})}:=\beta-\sum_{k^{\prime}=k_{1}+1}^{|v|}2^{-k^{\prime}}%
\beta_{k^{\prime}}.
\]
Then $\beta_{k^{\prime}}^{(k)}=\beta_{k^{\prime}}$ ($k^{\prime}\leq k_{1}$)
and $\beta_{k^{\prime}}=0$ ($k^{\prime}>k_{1}$). Therefore, $\beta_{k_{1}}$
can be decoded using the identities
\begin{align*}
\left\vert \zeta_{\beta_{k_{1}}}\right\rangle  &  =\pm(R_{\pi\beta^{(k_{1})}%
})^{2^{k_{1}}}\left\vert \zeta_{0}\right\rangle ,\\
R_{\pi\beta^{(k_{1})}}  &  =(R_{-\pi2^{-|v|}})^{\sum_{k^{\prime}=|v|}%
^{k_{1}+1}2^{|v|-k^{\prime}}\beta_{k^{\prime}}}R_{\pi\beta}.
\end{align*}
\ \ 

There are three counter tracks in the tape, corresponding to the variables
$k_{1}$, $k_{2}$, and $k_{3}$. Also there is a register in the finite control
that stores variable $w$ ($=w_{0}$ or $w_{1}$), which is initially $w_{0}$.
$k_{1}$ indicates that the algorithm is working at the $k_{1}$-th digit of
$\beta_{k_{1}}$, and it runs from $n^{\prime}$ to $0$. In each loop, we
multiply $R_{\pi\beta}$ for $2^{k_{1}}$ times. $k_{2}$ indicates the times of
application, so it runs from $1$ to $2^{k_{1}}$ in each loop. If $k_{1}=|v|$,
that is detected by (\ref{|v|}), rewrite $w_{0}$ to $w_{1}$. Hereafter,
between the $k_{2}$-th and $k_{2}+1$-th application of $R_{\pi\beta}$,
$R_{-\pi2^{-|v|}}$ is applied for $\sum_{k^{\prime}=|v|}^{k_{1}+1}%
2^{|v|-k^{\prime}}\beta_{k^{\prime}}$ times. $k_{3}$ indicates the times of
application of $R_{-\pi2^{-|v|}}$, and runs from $1$ to $\sum_{k^{\prime}%
=|v|}^{k_{1}+1}2^{|v|-k^{\prime}}\beta_{k^{\prime}}$ in each sub-loop. Here,
upon single application of $R_{\pi\beta}$ ($R_{-\pi2^{-|v|}}$, resp.),
$\mathcal{H}^{Q_{u}}$ changes from $q_{a,1}$ to $q_{b,1}$ (from $q_{a,2}$ to
$q_{b,2}$, resp.).

\begin{tabular}
[c]{l}%
Set $w:=w_{0}$, $\mathcal{H}^{Q_{in}}$ to $\left\vert \zeta_{0}\right\rangle
$, $k_{1}:=n^{\prime}$.\\
For $k_{1}=n^{\prime}$ to $1$ do\\
\ \ \ If $w=w_{0}$ do\\
\ \ \ \ \ \ \ For $k_{2}=1$ to $2^{k_{1}}$ do\\
\ \ \ \ \ \ \ \ \ \ \ Apply $R_{\pi\beta}$ to $\mathcal{H}^{Q_{in}}$. \ \\
\ \ \ \ \ \ \ If the content of $\mathcal{H}^{Q_{in}}$ is $\zeta_{1}$, output
$|v|=k_{1}$. Set $w:=w_{1}$.\\
\ \ \ \ \ \ Decrease $\ k_{1}$ by $1$, and go to the next loop\\
\ \ \ If $w=w_{1}$ do\ \ \ \ \\
\ \ \ \ \ \ \ For $k_{2}=1$ to $2^{k_{1}}$ do\\
\ \ \ \ \ \ \ \ \ \ \ \ \ Apply $R_{\pi\beta}$ to $\mathcal{H}^{Q_{in}}$.\\
\ \ \ \ \ \ \ \ \ \ \ \ \ For $k_{2}=1$ to $\sum_{k^{\prime}=|v|}^{k_{1}%
+1}2^{|v|-k^{\prime}}\beta_{k^{\prime}}$, apply $R_{-\pi2^{-|v|}}$ to
$\mathcal{H}^{Q_{in}}$.\\
\ \ \ \ \ \ \ \ \ \ \ \ \ Copy the content of $\mathcal{H}^{Q_{in}}$ to a
cell, and refresh it to $\zeta_{0}$.\\
\ \ \ \ \ \ \ \ \ \ \ \ \ Decrease $\ k_{1}$ by $1$, and go to the next loop.
\end{tabular}

Observe that the superposition occurs only in $\mathcal{H}^{Q_{in}}$, since
the content of the qubit is copied to other parts of the system only if it is
not in superposition. Observe also that the configuration at the two different
steps differ not only in $\mathcal{H}^{Q_{in}}$, but also in other parts which
are not in superposition. To see this, suppose the $j$-th and $j^{\prime}$-th
steps do not differ in either $k_{1}$, $k_{2}$, or $k_{3}$ -track of the tape,
but differ in the state of $\mathcal{H}^{Q_{in}}$. So at the $j$-th step
$R_{\pi\beta}$ or $R_{-\pi2^{-|v|}}$ had been applied but the counter $k_{2}$
or $k_{3}$ is not yet increased, and at the $j^{\prime}$-th step a rotation is
not yet applied. So the coalification at the $j$-th and $j^{\prime}$-th step
differ in $Q_{u}$ at least.

Therefore, define
\[
\left\vert j;\boldsymbol{x}\right\rangle \left\langle j;\boldsymbol{x}%
\right\vert :=\mathrm{tr}_{\otimes_{i=1}^{L+1}\mathcal{H}_{i}^{Q_{in}}}%
U^{j-1}\left\vert \boldsymbol{x}\right\rangle \left\langle \boldsymbol{x}%
\right\vert (U^{\dagger})^{j-1}.
\]
Then $\left\langle j;\boldsymbol{x}\right.  \left\vert j^{\prime
};\boldsymbol{x}\right\rangle =0$ if $j\neq j^{\prime}$, and Lemma
\ref{lem:dephase-3} is verified as well. So we can apply the arguments in the
proof of the second main lemma to the reduced density $\,\mathrm{tr}%
_{\mathcal{H}^{Q_{in}}}\,\rho^{L}(t)$, leading to the analogous statement.

\end{document}